\documentclass[journal, onecolumn]{IEEEtran}

\usepackage[latin1]{inputenc}
\usepackage{amsmath}
\usepackage{graphicx}
\usepackage{amssymb}
\usepackage{verbatim}
\usepackage{epsfig}
\usepackage[labelfont=bf]{caption}
\usepackage{subcaption}
\usepackage{enumerate}

\usepackage{cite}
\usepackage{epstopdf}

\newtheorem{remark}{Remark}
\newtheorem{proposition}{Proposition}

\newtheorem{proof}{Proof}

\newcommand{\col}{\mbox{col}}

\def\calg{{\cal G}}

\def\L2e{{\cal L}_{2e}}

\def\rea{\mathbb{R}}

\def\adj{\mbox{adj}}

\def\begequarr{\begin{eqnarray}}
\def\endequarr{\end{eqnarray}}
\def\begequarrs{\begin{eqnarray*}}
\def\endequarrs{\end{eqnarray*}}
\def\begarr{\begin{array}}
\def\endarr{\end{array}}
\def\begequ{\begin{equation}}
\def\endequ{\end{equation}}
\def\begequs{\begin{equation*}}
\def\endequs{\end{equation*}}
\def\lab{\label}
\def\begdes{\begin{description}}
\def\enddes{\end{description}}
\def\begenu{\begin{enumerate}}
\def\begite{\begin{itemize}}
\def\endite{\end{itemize}}
\def\endenu{\end{enumerate}}

\def\lef[{\left[\begin{array}}
\def\rig]{\end{array}\right]}

\def\begcen{\begin{center}}
\def\endcen{\end{center}}
\def\begrem{\begin{remark}\rm}
\def\endrem{\end{remark}}

\def\calm{{\cal M}}

\def\calf{{\cal F}}

\def\call{{\cal L}}
\def\calm{{\cal M}}

\def\calf{{\cal F}}

\def\rea{\mathbb{R}}

\usepackage{color}

\begin{document}

\title{Identification of Photovoltaic Arrays' Maximum Power Extraction Point via Dynamic Regressor Extension and Mixing}

\author{Anton~Pyrkin$^{1,*}$ , 
Fernando Mancilla-David$^2$,\\
Romeo~Ortega$^3$,
Alexey~Bobtsov$^1$, 
Stanislav~Aranovskiy$^{4,1}$

\thanks{$^{*}$ Anton~Pyrkin is the corresponding author, a.pyrkin@gmail.com}
\thanks{$^{1}$ Anton~Pyrkin, Alexey~Bobtsov and Stanislav~Aranovskiy are with the Department of Control Systems and Informatics, ITMO University, St. Petersburg, Russia.}
\thanks{$^{2}$ Fernando Mancilla-David is with the Department of Electrical Engineering, University of Colorado, Denver, Colorado 80217, USA.}
\thanks{$^{3}$ Romeo~Ortega is with the Laboratoire des Signaux et Syst\`emes, CNRS-SUPELEC, Plateau du Moulon, 91192, Gif-sur-Yvette, France.}
\thanks{$^{4}$ Stanislav~Aranovskiy is with the NON-A team, INRIA-LNE, Parc Scientifique de la Haute Borne 40, avenue Halley Bat.A, Park Plaza, 59650 Villeneuve d'Ascq, France.}
}

\maketitle

\begin{abstract}
This paper deals with the problem of identification of photovoltaic arrays' maximum power extraction point---information that is encrypted in the current-voltage characteristic equation. We propose a new parameterisation  of the classical five parameter model of this function that, combined with the recently introduced identification technique of dynamic regressor extension and mixing, ensures a fast and accurate estimation of all unknown parameters. A concavity property of the current-voltage characteristic equation is then exploited to directly identify the desired voltage operating point. Realistic numerical examples via computer simulations are presented to assess the performance of the proposed approach.
\end{abstract}

\section{Introduction}

In this paper we are interested in the optimisation of the operation of photovoltaic (PV) cell arrays, which in most applications reduces to the extraction of its maximum power. This corresponds to operation of the PV array at a specific voltage where the static current--voltage (IV) characteristic attains its maximum---known in the literature as maximum power point (MPP). Clearly, knowledge of the IV curve is essential for the identification of the MPP.  Although the data sheet provided by the manufacturer contains some information about the IV curve it is well--known  that its shape depends on the temperature and the solar irradiance, therefore it is of interest to dispose of an {\em on--line} estimator of the graph that can track its variations and relies only on the classically available measurements, {\em i.e.}, current and voltage.

The classical IV curve model \cite{masters} depends on five parameters that enter nonlinearly into the function. This stymies the application of standard identification techniques \cite{lju} to estimate these parameters. The first contribution of this paper is to propose a new parameterisation of the IV characteristic, which combined with the recently introduced dynamic regressor extension and mixing (DREM) identification technique  \cite{DREM}, generates {\em consistent estimates} of these parameters. A second contribution is the introduction of an adaptive observer of the MPP extraction {\em voltage} that, applying certainty equivalence to the estimated IV curve model and exploiting a key concavity property of the power graph, directly estimates its maximum point.

The MPP tracking is a well-known problem \cite{keyhani, masters} and may be solved by an extremum seeking approach (see, {\em e.g.} \cite{krstic}). The proposed non-invasive identification-based MPP extraction approach should be contrasted with the extremum seeking-based techniques that require the injection of external signals that perturb the normal operation of the PV array and rely on the generation of undesirable oscillations around the MPP.    

\section{Current--Voltage Characteristic Model}

The circuital model for a single PV cell and its generalization to a number of cells in series is well established in terms of a current source, an anti parallel diode, a series resistance and a shunt resistance \cite{masters}. To be able to treat also cells connected in parallel we follow here the modified circuit  introduced in \cite{tian2012}, which replaces the anti parallel diode by an externally controlled current source, as illustrated in Fig. \ref{fig_pv}. 

\begin{figure}[ht]
\centering
\includegraphics[width=0.8\textwidth]{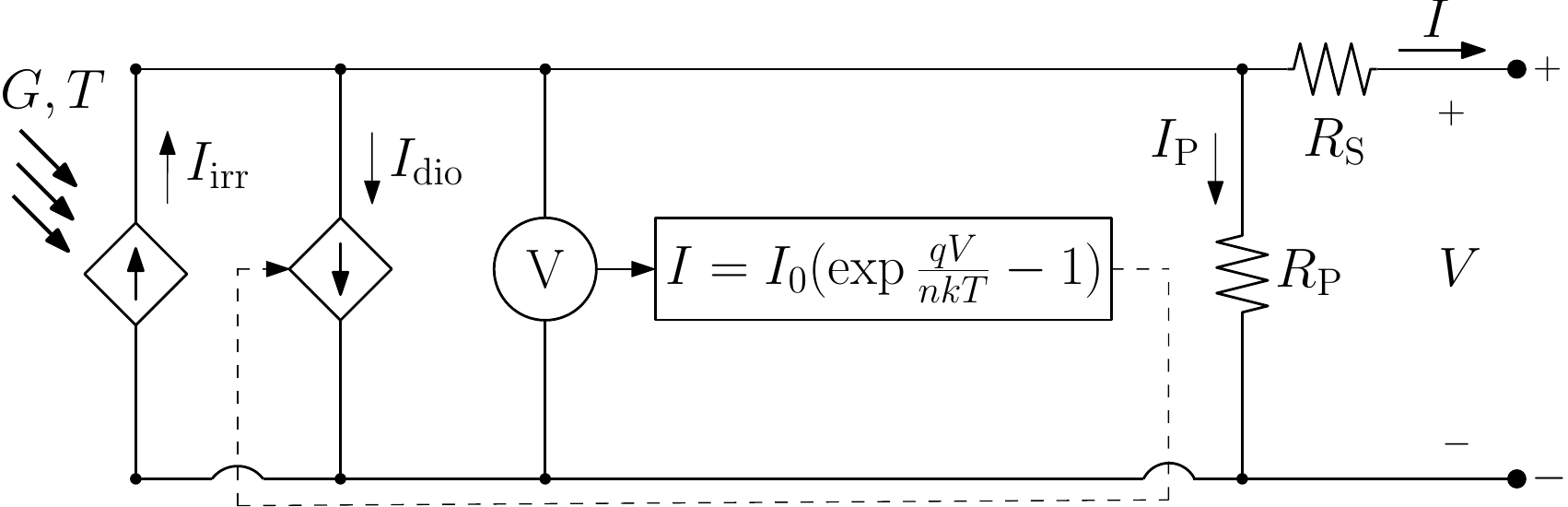}
\caption{Equivalent circuit for a single PV cell}
\vspace{4mm}
\label{fig_pv}
\end{figure}

The output current of the PV cell of Fig. \ref{fig_pv} may be expressed as,
\begin{align*}
I = I_\mathrm{irr} - I_\mathrm{dio} - I_\mathrm{p},
\end{align*}
where $I_\mathrm{irr}$ is the photo current or irradiance current generated when the cell is exposed to sunlight. $I_\mathrm{dio}$ is the current flowing through the anti-parallel diode and induces the nonlinear characteristics of the PV cell. $I_\mathrm{p}$ is a shunt current due to the shunt resistor $R_{P}$ branch. Substituting relevant expressions for $I_\mathrm{dio}$ and $I_p$ we get
\begin{align}
\lab{current1}
I = I_\mathrm{irr} - I_\mathrm{0} \left[\exp\left({q(V+IR_\mathrm{S})\over nkT}\right)-1\right]-{V+IR_\mathrm{S}\over R_\mathrm{P}},
\end{align}
where $q = 1.602\times 10^{-19} C$ is the electron's electric charge, $k = 1.3806503\times 10^{-23} J/K$ is the Boltzmann constant, $T$ is the temperature of the cell, $I_\mathrm{0}$ is the diode saturation current or cell reverse saturation current, $n$ is the ideality factor or the ideal constant of the diode, and $R_\mathrm{S}$ and $R_\mathrm{P}$ represent the series and shunt resistance, respectively \cite{masters}.

A PV power plant will usually contain a large number of cells in series, say $N_\mathrm{S}$, and in parallel, say $N_\mathrm{P}$, in order to match a particular level of voltage/power rating. For instance, if $0.5$ V$/2$ W PV cells are used as a building block to realize a $450$ kW power plant, $N_S$ and $N_P$ would be in the order of $1300$ and $180$, respectively. Those cells would in practice be grouped in panels---for example a panel may contain two parallel--connected strings made up of $36$ series-connected cells each. 

The model of a single cell is generalized to an arbitrary number $N_\mathrm{S} \times N_\mathrm{P}$ of cells connected in series and parallel to form an array giving the following final form \cite{tian2012} to the IV curve
\begin{align}
\lab{current2}
I &= N_\mathrm{P} I_\mathrm{irr} - N_\mathrm{P} I_\mathrm{0} \left[\exp\left({q(V+I{N_\mathrm{S}\over N_\mathrm{P}}R_\mathrm{S})\over N_\mathrm{S}nkT}\right)-1\right]
-{V+I{N_\mathrm{S}\over N_\mathrm{P}}R_\mathrm{S}\over {N_\mathrm{S}\over N_\mathrm{P}}R_\mathrm{P}}.
\end{align}

Identification of the IV curve \eqref{current2} is stymied by the fact that the five unknown parameters ($I_\mathrm{irr}$, $I_\mathrm{0}$, $R_\mathrm{S}$, $R_\mathrm{P}$, $n$) enter {\em nonlinearly} into the function. Moreover, they change with the solar irradiance $G$, the cell temperature $T$, and certain reference parameters, namely, $I_\mathrm{irr,ref}$, $I_\mathrm{0,ref}$, $R_\mathrm{S,ref}$, $R_\mathrm{P,ref}$ and $n_\mathrm{ref}$, are measured at standard reference conditions, usually at $G_\mathrm{ref}=1000$ W/m$^2$ and $T_\mathrm{ref}=25^\mathrm{o}$C. See \cite{carrasco2014} for more details on these dependence and \cite{tian2012} where a methodology to derive the reference parameters  from information available on manufacturers' datasheets is proposed.

The IV curves described by \eqref{current2}  characterise the (static) performance of a PV array. Typical curves are illustrated in Fig. \ref{fig:IVcurves}. The figure shows the MPP, which is of special interest in applications as it allows to operate the PV array at its maximum efficiency for a given value of $G$ and $T$.  This is achieved calculating the voltage corresponding to the MPP, which may be utilized as a reference voltage for the power converter that links the PV array with the load or the power grid.

If $G$ and $T$ are measurable it is possible to use the reference parameters to track the MPP voltage. To avoid the need of additional measurements we propose in this paper to generate an on-line estimator of the IV curve that can, in principle, track the parameter variations using only measurements of voltage and current. Towards this end, a new parameterisation of the IV characteristic \eqref{current2} and a novel on--line estimator of these new parameters is proposed in the paper.  The estimator is designed using the DREM technique, which was recently proposed in  \cite{DREM} and is briefly described in Section \ref{sec5}. Using the estimated parameters we propose in Section \ref{sec7} an adaptive observer of the MPP extraction {\em voltage} that directly estimates the maximum point of the power graph.

\begin{figure}[ht]
\centering
\includegraphics[width=0.8\textwidth]{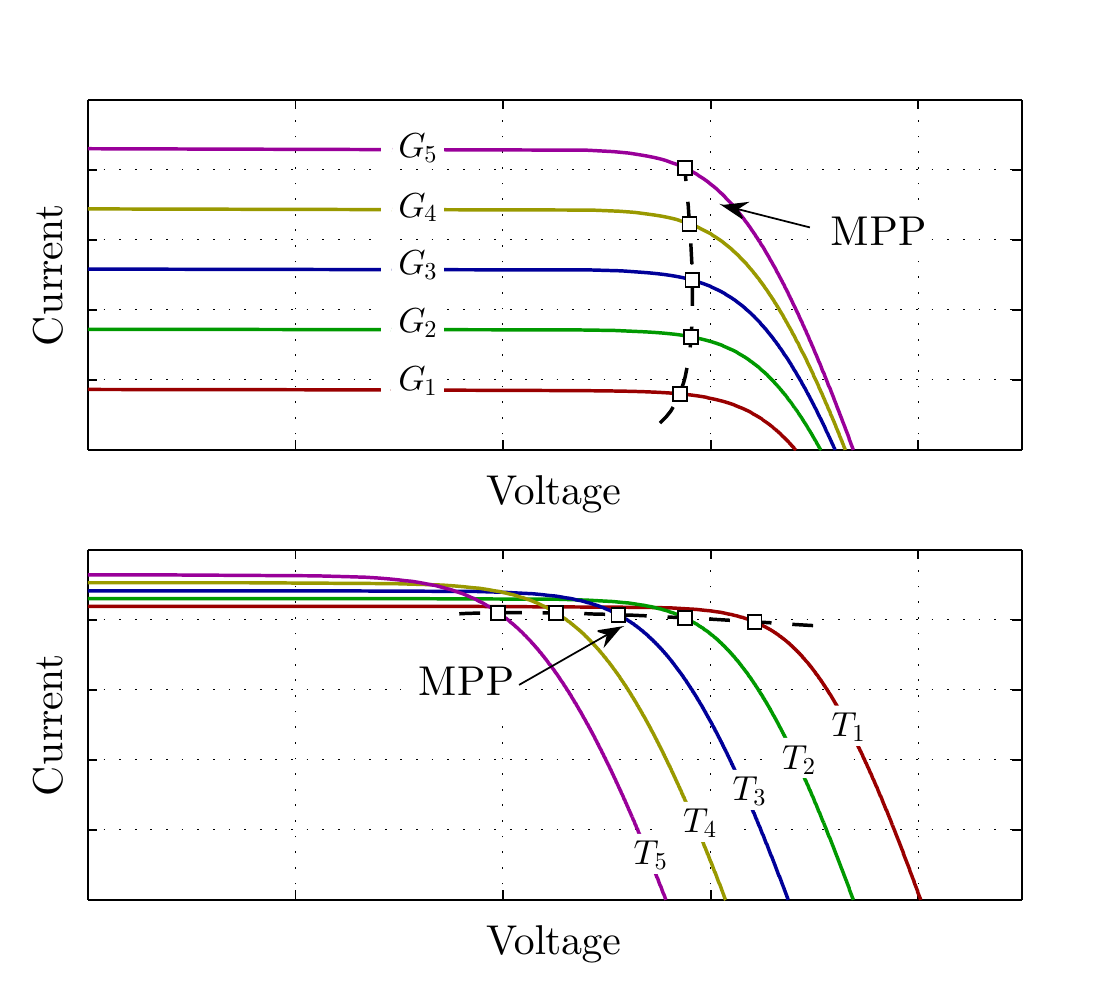}
\caption{Typical IV  performance curves. Top and bottom plots illustrate curves for different values of irradiance ($G_1<G_2<G_3<G_4<G_5$) and temperature ($T_1<T_2<T_3<T_4<T_5$), respectively.}\label{fig:IVcurves}
\end{figure}

\section{Formulation of Parameter Identification Problem}
\lab{sec3}

For the sake of simplicity, and without loss of generality, we assume the PV array is utilized to charge a battery bank through a boost dc/dc converter. Fig. \ref{fig:PowerStage} illustrates the architecture. 

\begin{figure}[ht]
\centering
\includegraphics[width=0.8\textwidth]{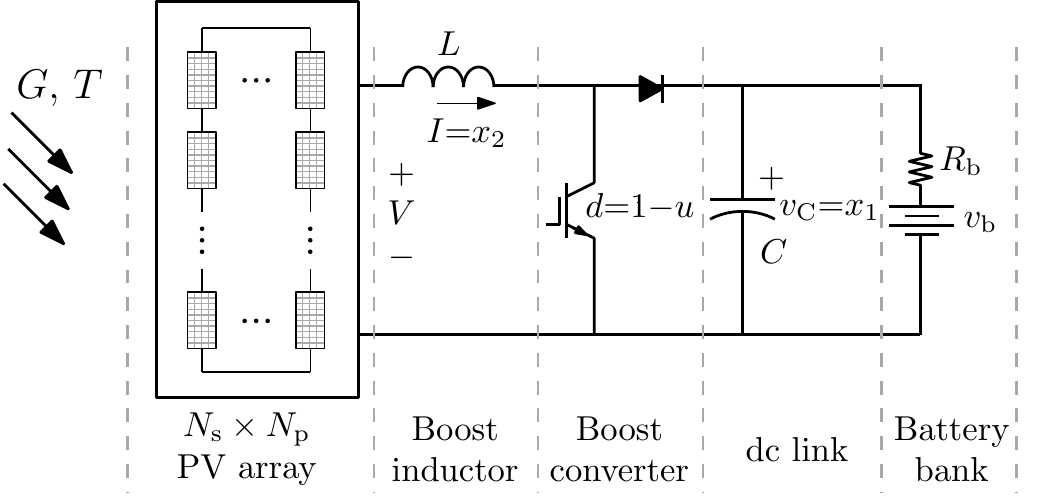}
\caption{Schematic of PV array charging a battery bank through a boost dc/dc converter. The voltage $V$ across the PV array is controlled by adjusting the transistor duty cycle, $d$.}\label{fig:PowerStage}
\end{figure}

The state space model for the system of Fig. \ref{fig:PowerStage} is given by the following---nonlinear and nonlinearly parameterised---differential--algebraic equations
\begin{eqnarray}
 \lab{x1}
      C\dot v_C & = & u\,I - {1\over R_\mathrm{b}}(v_C-v_\mathrm{b}),\\
 \lab{x2}
      L\dot I & = & -u v_C + V,\\
 \lab{I}
      I & = & a_1 - a_2\,\left(e^{a_3\,(V+a_4\,I)}-1\right)-a_5(V+a_4\,I),
\end{eqnarray}
with the following definitions:
\begite
\item[(i)]  $(v_C,I,V,u)$ are positive measurable signals, with $u=1-d$, where $d$ is the duty cycle that is chosen by the designer.
\item[(ii)] $L,C,R_\mathrm{b},v_\mathrm{b}  >0$ are known constants.
\item[(iii)] $a:=\col(a_1,...,a_5)$ are positive {\em unknown parameters}.
\endite

We make the important observation that, if the parameters $a$ are {\em known}, it is possible to compute the MPP voltage by evaluating the derivative (with respect to the voltage $V$) of the extracted power $VI$ and setting it equal to zero. In view of this observation the first step for the identification of the MPP voltage is the solution of the following parameter identification problem.

\vspace{6mm}
\noindent{\em Problem formulation.} 
Consider the system   \eqref{x1}--\eqref{I}, with $u$ free to the designer and such that all signals are bounded.  From measurements of $(v_C,I,V,u)$,  generate {\em consistent} on--line estimates of all  parameters $a$---that is, denoting the estimates as $\hat a$, they should verify
\begequs
\lim_{t\to \infty} \left\| a-\hat a(t)\right\|=0,
\endequs
with $\|\cdot\|$ the Euclidean norm.
 
\begrem
It can be observed from Fig. \ref{fig:PowerStage} that by changing the transistor duty cycle one can adjust the voltage $V$ to ``sweep'' through the PV array's IV curve and dynamically generate a graph similar to the ones in Fig. \ref{fig:IVcurves}---for a given value of $G$ and $T$. Given the complicated dynamics of the system and the large number  of unknown parameters, it is clear that this excitation step is essential to generate consistent estimates.
\endrem

\begrem
An architecture similar to the one studied here has been considered in the interesting paper \cite{mezetal} where an adaptive PI controller, like the one proposed in \cite{heretal},  is designed to regulate the voltage.  It should be underscored that the IV curve model considered in  that paper is a much simpler three--parameter relation (equation (1)  in   \cite{mezetal}) that is valid only for a single cell PV. Moreover, due to the nonlinear dependence of the parameters, only two of them are actually identified with the third one assumed known.  
\endrem

\begrem
If the temperature $T$ is {\em known}, it is also possible to obtain the physical parameters  ($I_\mathrm{irr}$, $I_\mathrm{0}$, $R_\mathrm{S}$, $R_\mathrm{P}$, $n$) from knowledge of $a$. More precisely, 
there exists a bijective mapping $\calm_\mathrm{T}:\rea_+^5 \mapsto \rea_+^5$, parameterised by the temperature $T$, such that 
$$
(I_\mathrm{irr},I_\mathrm{0},R_\mathrm{S},R_\mathrm{P},n)=\calm_\mathrm{T}(a).
$$
Clearly, the mapping is defined as
$$
 \left[
\begin{array}{c}
I_\mathrm{irr}\\
I_\mathrm{0}\\
n\\
R_\mathrm{S}\\
R_\mathrm{P}
\end{array}
\right]=\left[
\begin{array}{c}
{1 \over N_\mathrm{P}}a_1\\ \vspace{1mm}
{1 \over N_\mathrm{P}}a_2 \\ \vspace{1mm}
{q\over N_\mathrm{S} kT} {1 \over a_3}\\ \vspace{1mm}
{N_\mathrm{P} \over N_\mathrm{S}} a_4 \\ \vspace{1mm}
{N_\mathrm{P}\over N_\mathrm{S}}{1 \over a_5}
\end{array}
\right],
$$
where we underscore the presence of $T$ in third right hand term.
\endrem

\section{Model Reparameterization}
\lab{sec4}

An essential step in the estimator design is to propose a {\em linear reparametrisation} of the IV function, as shown below.

\begin{proposition}\em
The system \eqref{x1}--\eqref{I} admits a linear regression model  
\begequ
 \lab{regmod}
y  =  \Omega^\top\theta+\varepsilon_t,
\endequ
where $ \Omega$ and $y$ are measurable signals, $\theta:=\col(\theta_1, \dots, \theta_5) $ is a vector of unknown, positive parameters and $\varepsilon_t$ is an exponentially decaying term. Moreover, there exists a surjective mapping $\calf:\rea_+^5 \mapsto \rea_+^5$ such that 
$$
\theta=\calf(a).
$$

\end{proposition}
\begin{proof}
First, rewrite equation \eqref{I} in the more compact form
$$
      I+a_5a_4\,I  =  a_1 + a_2 - a_2\,e^{a_3\,(V+a_4\,I)}-a_5\,V.
$$
Hence
\begequ
 \lab{I_base}
      I  =   b_1 - b_2\,e^{a_3\,(V+a_4\,I)}-b_3\,V=:F(V,I),
\endequ
where, to simplify the expressions, we introduced the dummy parameters
\begin{align}
\lab{b_i}
b_1:={a_1 + a_2\over 1+a_5a_4},\quad b_2:={a_2\over 1+a_5a_4}, \quad b_3:={a_5\over 1+a_5a_4}.
\end{align}
and, for later reference, we have defined the function $F(V,I)$.

Differentiating \eqref{I_base} with respect to time and using \eqref{x2} we get 
\begin{eqnarray}
      \dot I & = & - b_2a_3\,(\dot V+a_4\,\dot I)\,e^{a_3\,(V+a_4\,I)}-b_3\,\dot V\nonumber\\
      & = &  a_3\,(\dot V+a_4\,\dot I)\,(I-b_1+b_3\,V)-b_3\,\dot V\nonumber\\
      & = &  a_3\, \dot V I -a_3b_1\,\dot V+a_3b_3 \, V\dot V +a_3a_4\,I\dot I-a_3a_4b_1\,\dot I \nonumber \\ &\quad&
      +a_3a_4b_3\,V\dot I-b_3\,\dot V\nonumber
      \\
      & = & a_3\, \dot V I -(a_3b_1+b_3)\,\dot V+a_3b_3 \, V\dot V +a_3a_4\,I\dot I
      \nonumber\\&\quad&
      -a_3a_4b_1\,\dot I+a_3a_4b_3\,V\dot I,\nonumber
\end{eqnarray}
yielding
\begin{eqnarray}
 \lab{I_dot}
\dot I & = & { a_3 \over 1+ a_3a_4b_1}\, \dot V I -{a_3b_1+b_3 \over 1+ a_3a_4b_1}\,\dot V+{a_3b_3 \over 1+ a_3a_4b_1} \, V\dot V 
\nonumber\\&\quad&
+{a_3a_4 \over 1+ a_3a_4b_1}\,I\dot I-{a_3a_4b_3\over 1+ a_3a_4b_1}\,V\dot I \nonumber\\
& = & \theta_1\, \dot V I -\theta_2\,\dot V+\theta_3 \, V\dot V +\theta_4\,I\dot I-\theta_5\,V\dot I\nonumber\\
& = & \theta_1\, \dot V I \!-\!\theta_2\,\dot V\!+\!{\theta_3 \over 2}  {d\over dt} {V^2} \!+\!{\theta_4\over 2} {d\over dt} {I^2}\!-\!\theta_5\,V\dot I,
\end{eqnarray}
where we defined the positive constants
\begin{align}
\theta_1&:={ a_3 \over 1\!+\! a_3a_4b_1},
\theta_2:={a_3b_1+b_3 \over 1\!+\!a_3a_4b_1},
\theta_3:={a_3b_3 \over 1\!+\! a_3a_4b_1},
\theta_4:={a_3a_4 \over 1+ a_3a_4b_1},
\theta_5:={a_3a_4b_3\over 1+ a_3a_4b_1}.
\lab{theta45}
\end{align}

Now, apply to \eqref{I_dot} an operator of the form ${\lambda\over p+\lambda}$ where $p:={d \over dt}$ and $\lambda>0$ is a designer chosen constant, to get
\begin{eqnarray}
 \lab{I_dot2}
\left[\lambda\over p+\lambda\right]\dot I & = & \theta_1\, \left[\lambda\over p+\lambda\right](\dot V I) -\theta_2\,\left[\lambda\over p+\lambda\right]\dot V+{\theta_3\over 2} {\left[\lambda\over p+\lambda\right]}\,{d\over dt} V^2\nonumber\\
&\quad&  +\,{\theta_4\over 2}{\left[\lambda\over p+\lambda\right]}\,{d\over dt} I^2
-\theta_5\,{\left[\lambda\over p+\lambda\right]}(V\dot I)+\varepsilon_t,
\end{eqnarray}
where $\varepsilon_t$ is the generic notation for an exponentially decaying term stemming from the filters initial conditions.

To generate the regressor vector $\Omega$ we introduce five linear filters\footnote{Notice that the right hand side of \eqref{xi_45} is computable without differentiation using the right hand side of \eqref{x2}.}
\begin{align}
\lab{xi_123}
\xi_1&=\left[\lambda \over p+\lambda\right]\, I,
\quad
\xi_2=\left[- \lambda \over p+\lambda\right]\, V,
\quad
\xi_3={1\over 2}\left[\lambda \over p+\lambda\right]\, V^2,\\
\label{xi_45}
\xi_4&= {1\over 2}\left[\lambda \over p+\lambda\right]\,I^2,
\quad
\xi_5=\left[-\lambda \over p+\lambda\right]\, (V \dot I).
\end{align}
To deal with the first right hand term of \eqref{I_dot2}, which contains the unmeasurable signal $\dot V$, we recall the Swapping Lemma (see, {\em e.g.}, \cite{sastry_bodson}), which applied to $\xi_5$ in \eqref{xi_45} yields
\begin{align}
\left[\lambda \over p+\lambda\right]\,(\dot V I) & =
I \left[\lambda \over p+\lambda\right] \dot V 
- \left[1 \over p+\lambda\right] \left( \dot I \left[\lambda \over p+\lambda\right] \dot V \right).
\lab{swap_lemma}
\end{align}
Replacing the filtered expressions above in  \eqref{I_dot2}, and using \eqref{swap_lemma}, we get
$$
\dot \xi_1  =  \theta_1 \Omega_1 +\theta_2\,\dot \xi_2+\theta_3 \, \dot\xi_3 +\theta_4\,\dot \xi_4+\theta_5\,\xi_5  +\varepsilon_t,
$$
where we defined the signal
$$
\Omega_1:=I \dot\xi_2 - \left[1 \over p+\lambda\right] \left( \dot I \dot\xi_2\right).
$$
Defining $y:=\dot \xi_1$ and the regressor vector
$$ 
\Omega:=\col(\Omega_1, \dot \xi_2, \dot\xi_3,\dot \xi_4, \xi_5),
$$
establishes the first claim.

The surjective mapping $\calf(\cdot)$ is defined replacing $b_1$ and $b_3$ of \eqref{b_i} in \eqref{theta45}:
\begequ
\begin{bmatrix} \theta_1 \\ \theta_2 \\ \theta_3 \\ \theta_4 \\ \theta_5 \end{bmatrix}={1\over 1+a_4a_5+a_3a_4(a_1+a_2)}\begin{bmatrix}a_3(1+a_4a_5)\\ a_5+a_3(a_1+a_2) \\ a_3a_5 \\ a_3a_4(1+a_4a_5) \\ a_3a_4a_5  \end{bmatrix}.\nonumber
\endequ
\end{proof}

\begrem
\lab{rem3}
It is easy to show that, if $\dot V$ is available for measurement, it is possible to construct a linear regression model similar to \eqref{regmod} that {\em does not} require the measurement of $v_C$, nor the knowledge of $L$. 
\endrem

\section{DREM Parameter Estimator}
\lab{sec5}

The next step is to propose an estimator for the parameters $\theta$ using the linear regression  \eqref{regmod}. Towards this end, we will use here the DREM technique proposed in \cite{DREM}. Our motivation to use DREM stems from the fact that convergence is established without imposing on the signals the stringent persistent excitation conditions which are required in standard gradient (or least--squares) estimators \cite{lju}. 

Following the DREM methodology we define an extended regression mapping introducing the scaled, delayed signals
\begequ
y_{f_j}(t)  :=  {\beta}\,y(t-d_j),\;\Omega_{f_j}(t)  :=  {\beta}\,\Omega(t-d_j),
\label{d_operator}
\endequ
where $d_j \in \rea_+,\;j=1,\dots,4$ {and $\beta \in \rea_+$ is an auxiliary constant that, as will be shown below, is a scaling factor that ameliorates the {\em numerical conditioning} in the DREM calculations}. From  \eqref{regmod} it is clear that
$$
y_{f_j} =\Omega^\top_{f_j} \theta,
$$
where here, and throughout the rest of the paper, we neglect the presence of the additive exponentially decaying terms $\varepsilon_t$. The interested reader is referred to \cite{DREM} for the analysis of the effect of these terms.

Piling up the original regression equation \eqref{regmod} with the $4$ new filtered regressions we construct the extended regressor system
\begequ
\label{YM}
Y_e = M_e \theta,
\endequ
where we defined $Y_e \in \rea^5,\;M_e \in \rea^{5 \times 5}$ as
\begequ
\lab{yama}
Y_e:=	\begin{bmatrix} y \\ y_{f_1} \\ y_{f_2} \\y_{f_3} \\ y_{f_{4}} \end{bmatrix},\;M_e:=\begin{bmatrix} \Omega^\top \\ \vspace{1mm} \Omega^\top_{f_1} \\ \vspace{1mm} \Omega^\top_{f_2} \\ \vspace{1mm} \Omega^\top_{f_3} \\ \vspace{1mm} \Omega^\top_{f_{4}} \end{bmatrix}.
\endequ
Premultiplying  \eqref{YM} by the adjunct matrix of $M_e$ we get $5$ scalar regressors of the form
\begequ
\lab{scareg}
Y_i = \Delta \theta_i
\endequ
with $i = 1,\dots 5$, where we defined the scalar function
\begequ
\lab{phi}
\Delta:=\det \{M_e\},
\endequ
and the vector $Y=\col(Y_1,\dots,Y_5) \in \rea^5$
\begequ
\lab{Y}
Y:= \adj\{M_e\} Y_e.
\endequ

The estimation of the parameters $\theta_i$ from the scalar regression form \eqref{scareg} can be easily carried out via
\begequ
\lab{decest}
\dot{\hat{\theta}}_i = -\gamma_i\Delta (\Delta \hat\theta_i - Y_i),
\endequ
with adaptation gains $\gamma_i>0$. From \eqref{scareg} it is clear that the latter equations are equivalent to
\[
\dot{\tilde{\theta}}_i = -\gamma_i \Delta^2 \tilde\theta_i,
\]
where $\tilde \theta_i:=\hat \theta_i - \theta_i$ are the parameter errors. Solving this simple scalar differential equation we conclude that
\begequ
\lab{equsta}
\lim_{t\to \infty} \tilde \theta_i(t)=0 \quad \Longleftrightarrow \quad \Delta(t) \notin \call_2.
\endequ

The derivations  above establish the following proposition.\\

\begin{proposition}\em
\lab{pro5.1}
Consider the $5$--dimensional linear regression \eqref{regmod}. Define the vector $Y_e$ and the matrix $M_e$ as given in \eqref{d_operator}--\eqref{yama}. The estimator \eqref{decest}, with $\Delta$ and $Y$ defined in \eqref{phi} and \eqref{Y}, respectively, verifies \eqref{equsta}.\\
\end{proposition}

\begin{figure*}[ht]
\centering
	\subcaptionbox{\label{fig41} Current $I(t)$}{\includegraphics[width=0.32\textwidth]{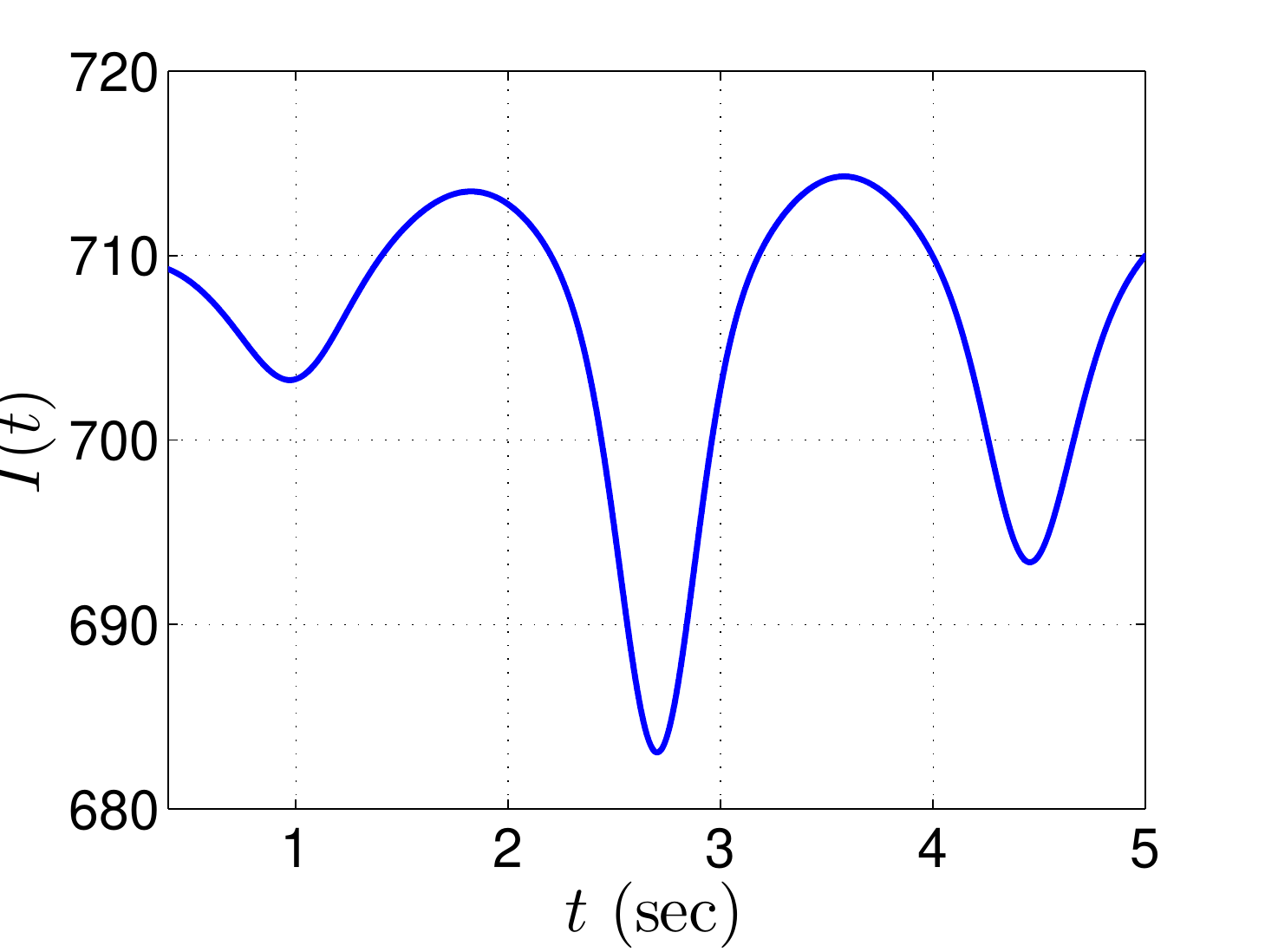}}
	\subcaptionbox{\label{fig42} Voltage $V(t)$}{\includegraphics[width=0.32\textwidth]{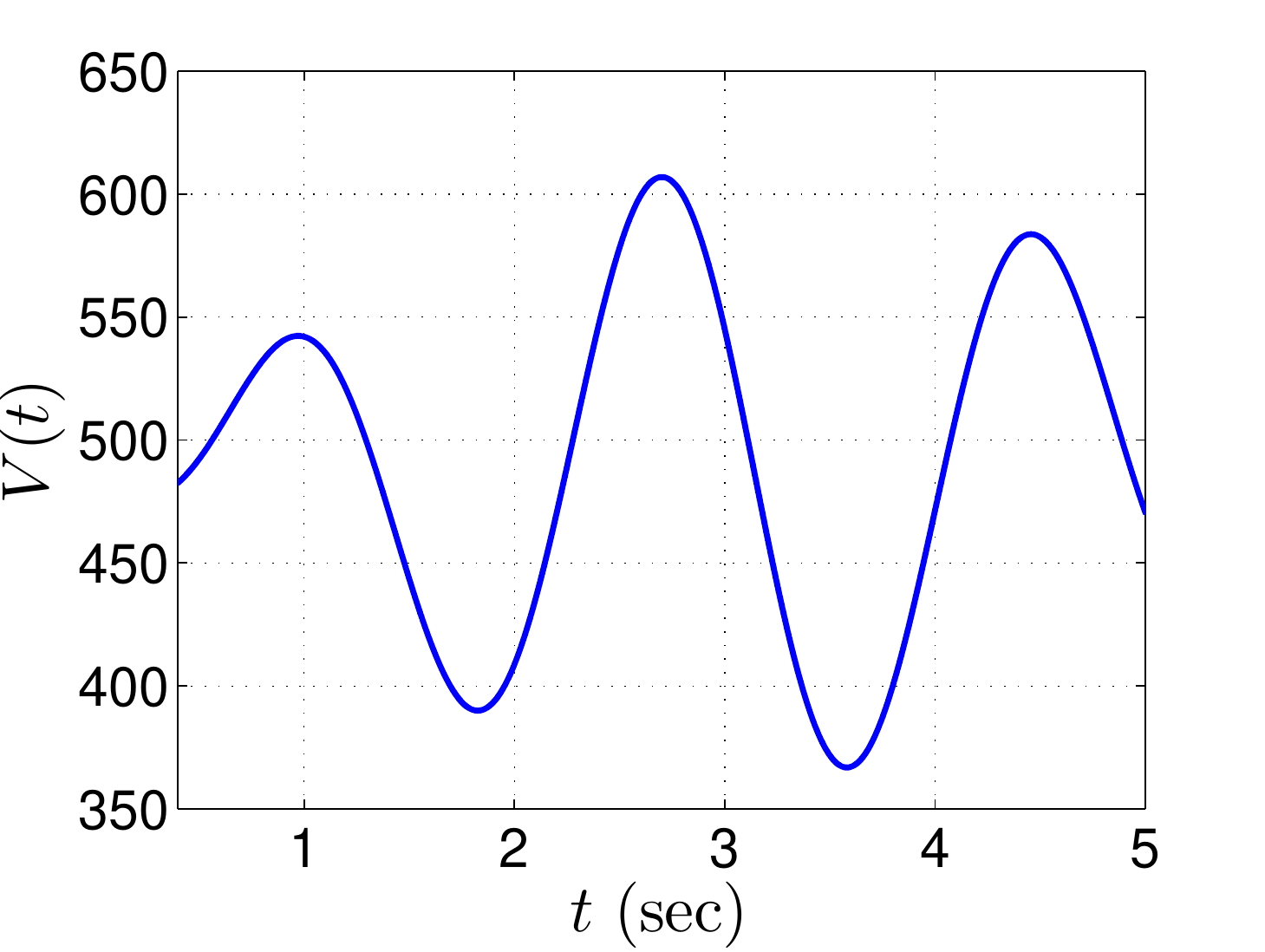}}
	\subcaptionbox{\label{fig43} Control $u(t)$}{\includegraphics[width=0.32\textwidth]{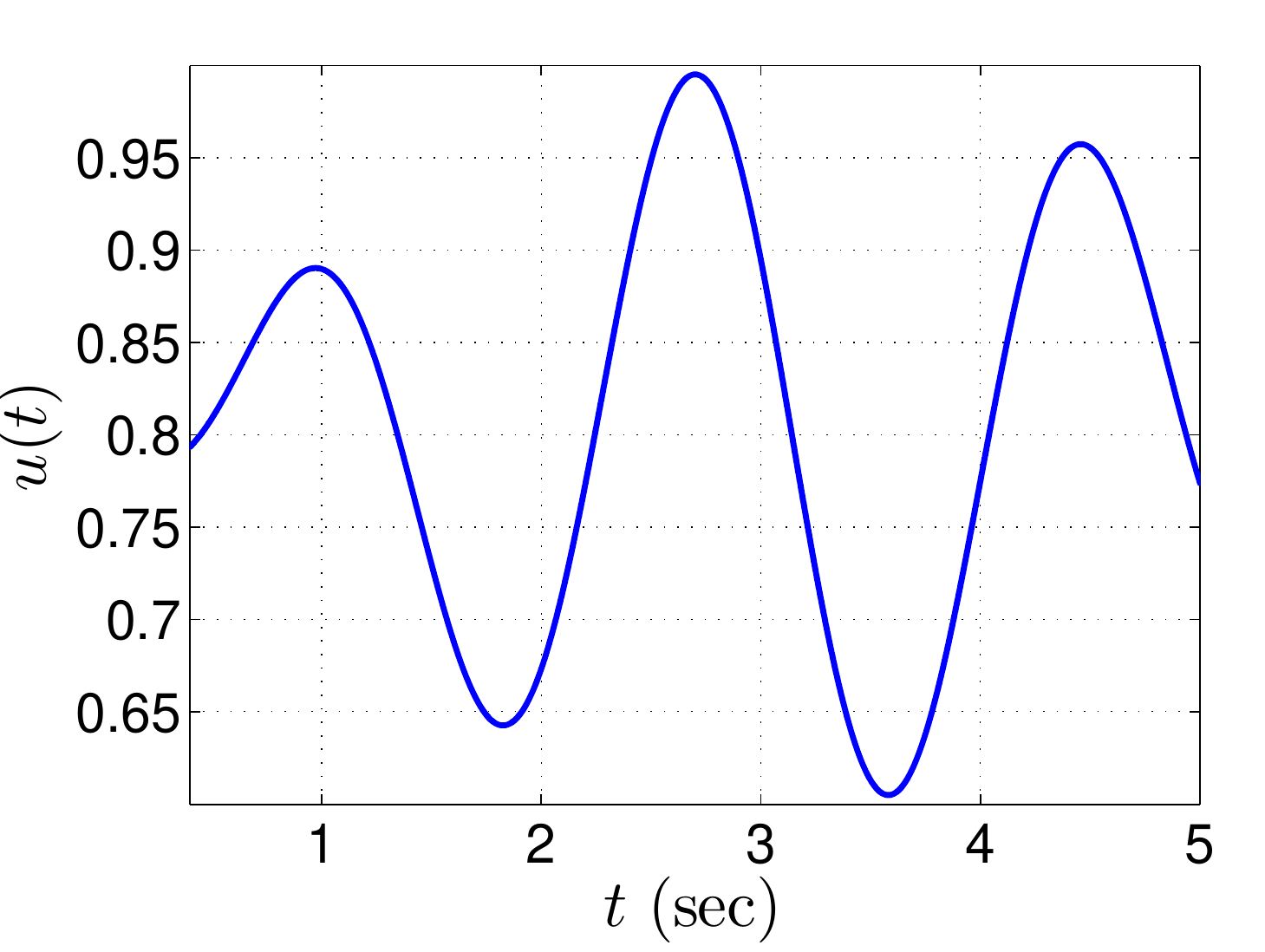}}
	\caption{\label{fig4} Current, voltage and control profiles.}
\end{figure*}

\begrem
In \cite{DREM} examples of regressors, which are not persistently exciting but satisfy the non--square integrability condition \eqref{equsta} are reported.
\endrem

\begrem
In the general formulation of DREM it is possible to use other linear operators different from the pure delay ones proposed here. The interested reader is referred to \cite{DREM} for further details on DREM.
\endrem

\section{A Simpler Parameterization}
\lab{sec6}

Interestingly, it is shown below that only four elements of $\theta$ are required to  reconstruct all the parameters $a$.
 
\begin{proposition}\em
\label{pro6.1}
There exists a surjective mapping $\calg:\rea^6 \mapsto \rea^5$ such that\footnote{A similar result may be established for $(\theta_1, \theta_2, \theta_4, \theta_5)$ but is omitted for brevity.}
$$
a=\calg(\theta_1, \theta_2, \theta_3, \theta_4,V,I).
$$
\end{proposition}

\begin{proof}
Using  \eqref{theta45} we get
\begin{align*}
a_4={\theta_4\over\theta_1},
\qquad
b_3={\theta_3\over\theta_1},
\end{align*}
and from the expressions for $\theta_1$ and $\theta_2$ we obtain
\begin{align*}
a_3={\theta_1(1-a_4b_3)\over 1+\theta_2a_4},
\qquad
b_1={b_3+\theta_2\over \theta_1(a_4b_3-1)},
\end{align*}
From \eqref{b_i} we have $a_5={b_3\over 1-a_4b_3}$ and
\begin{align}
\lab{a12}
a_{1}+a_2=b_1 (1+a_4a_5),
\end{align}
On the other hand, from \eqref{I} we have
$$
      -I + (a_1+a_2) -a_5 (V+a_4I)  =   a_2\,e^{a_3\,(V+a_4\,I)},
$$
from which, replacing \eqref{a12}, we get
$$
a_2=(-I +b_1 (1+a_4a_5) - a_5 (V+a_4I))\,e^{-a_3\,(V+ a_4\,I)}.
$$
With \eqref{a12} and the latter equation one can find the last parameter $a_1$. The final mapping $\calg$ takes the form
\begequarr
\lab{G_a1}
a_1&=&-\left[{\theta_1^2}I
+{\theta_1\theta_3}V
+{\theta_1^2(\theta_1\theta_2-\theta_3)\over
(\theta_3\theta_4-\theta_1^2)}\right]
{e^{{\theta_3\theta_4-\theta_1^2\over\theta_1-\theta_2\theta_4}
(V+{\theta_4\over\theta_1}I)}
\over\theta_3\theta_4-\theta_1^2}+{\theta_1^2(\theta_1\theta_2-\theta_3)\over
(\theta_3\theta_4-\theta_1^2)^2},\\
\lab{G_a2}
a_2&=&\left[{\theta_1^2}I
+{\theta_1\theta_3}V
+{\theta_1^2(\theta_1\theta_2-\theta_3)\over
(\theta_3\theta_4-\theta_1^2)}\right]
{e^{{\theta_3\theta_4-\theta_1^2\over\theta_1-\theta_2\theta_4}
(V+{\theta_4\over\theta_1}I)}
\over\theta_3\theta_4-\theta_1^2},\\
\lab{G_a345}
a_3&=&{\theta_3\theta_4-\theta_1^2\over\theta_2\theta_4-\theta_1},\qquad
a_4={\theta_4\over\theta_1},
\qquad
a_5={\theta_1\theta_3\over{\theta_1^2-\theta_3\theta_4}}.
\endequarr
completing the proof.
\end{proof}

\section{Direct Identification of the MPP Voltage}
\lab{sec7}

To compute the MPP voltage from the knowledge of the IV characteristic we can proceed as follows. First, we compute the power extracted from the PV array as
$$
      P  =  V\,I=VF(V,I),
$$
where we have used \eqref{I_base} to obtain the second equation. Now, define the function  
$$
h(V,I):={\partial P\over\partial V},
$$
that may be computed as
$$
      h(V,I) = F(V,I)+V\,{\partial F(V,I)\over\partial V}= b_1 - (b_2+b_2a_3V)\,e^{a_3\,(V+a_4\,I)}-2b_3\,V.
$$

At the MPP we have $h(V_{\tt MPP},I_{\tt MPP})=0$. Combining this identity with \eqref{I_base} evaluated at the MPP yields
\begequ
\lab{I_MPP}
I_{\tt MPP}= {2b_3\,V_{\tt MPP}-b_1\over 1+a_3V_{\tt MPP}}+b_1-b_3\,V_{\tt MPP}=:g(V_{\tt MPP}).
\endequ
The MPP voltage is then computed defining the function
\begequ
\lab{H}
H(V_{\tt MPP}):= h(V_{\tt MPP},g(V_{\tt MPP}))=  b_1 - (b_2+b_2a_3 V_{\tt MPP})\,e^{a_3\,\left(V_{\tt MPP}+a_4\,\left({2b_3\, V_{\tt MPP}-b_1\over 1+a_3 V_{\tt MPP}}+b_1-b_3\, V_{\tt MPP}\right)\right)}-2b_3\, V_{\tt MPP}
\endequ
and solving for $V_* \in \rea_+$ the implicit equation $H(V_*)=0$.\footnote{Notice that to avoid cluttering the notation we have relabelled the MPP voltage $V_*$.}  

The proposition below shows that the last step, which involves some far from obvious calculations, can be carried out in a recursive manner exploiting the concavity of the power function evaluated at the MPP current.

\begin{proposition}\em
\lab{prop_mpp}
Consider the function \eqref{H}. The adaptation algorithm 
\begin{align}
\lab{mpp_est}
\dot{\hat V}_{*}=\gamma_V\,H(\hat V_{*}).
\end{align}
guarantees $\lim_{t\rightarrow\infty}\hat V_{*}(t)=V_{*}$ for all constants $\gamma_V>0$ and all initial conditions $\hat V_*(0)>0$. 
\end{proposition}

\begin{proof}
The gist of the proof is to show that the function $H(\cdot)$ is {\em strictly decreasing} in the positive half line. This, together with the fact that $H(V_*)=0$, ensures that the differential equation \eqref{mpp_est} has an asymptotically stable equilibrium at $V_*$ that attracts all trajectories starting with  $\hat V_*(0)>0$. 

To prove the monotonicity lets compute the derivative of \eqref{H}  
\begin{align}
{dH(\hat V_*)\over d\hat V_{*}}&=
- b_2a_3\,e^{a_3\,\left(\hat V_{*}+a_4\,\left({2b_3\,\hat V_{*}-b_1\over 1+a_3\hat V_{*}}+b_1-b_3\,\hat V_{*}\right)\right)}\nonumber\\
&\quad - b_2a_3(1+a_3\hat V_{*})\,e^{a_3\,\left(\hat V_{*}+a_4\,\left({2b_3\,\hat V_{*}-b_1\over 1+a_3\hat V_{*}}+b_1-b_3\,\hat V_{*}\right)\right)}\times\nonumber\\&\quad\times\left(1-a_4b_3+a_4{2b_3(1+a_3\hat V_{*})-a_3(2b_3\hat V_{*}-b_1)\over (1+a_3\hat V_{*})^2}\right) -2b_3\nonumber\\
&=
- b_2a_3\,e^{a_3\,\left(\hat V_{*}+a_4\,\left({2b_3\,\hat V_{*}-b_1\over 1+a_3\hat V_{*}}+b_1-b_3\,\hat V_{*}\right)\right)}\nonumber\\
&\quad - b_2a_3(1+a_3\hat V_{*})\,e^{a_3\,\left(\hat V_{*}+a_4\,\left({2b_3\,\hat V_{*}-b_1\over 1+a_3\hat V_{*}}+b_1-b_3\,\hat V_{*}\right)\right)}\left({1\over1+a_4a_5}
+a_4{2b_3+a_3b_1\over (1+a_3\hat V_{*})^2}\right) \nonumber\\&\quad-2b_3<0,
\end{align}
where the inequality holds for all $\hat V_{*}>0$. This completes the proof.
\end{proof}

\begrem
In its adaptive implementation the function $\hat H(\hat V_{*})$ in \eqref{mpp_est} is evaluated using the estimates of the parameters ($\hat b_1$, $\hat b_2$, $\hat b_3$, $\hat a_3$, $\hat a_4$) computed with the DREM algorithm proposed in Section \ref{sec5}.
\endrem

\section{Simulation Results}
\lab{sec8}

To assess the performance of the proposed estimators realistic simulations were done for a PV system considered in \cite{carrasco2014, tian2012}, whose parameters are: $I_{irr,ref}=2.4207$, $I_{0,ref}=1.996\times 10^{-8}$, $R_{S,ref}=1.526\times 10^{-2}$, $R_{P,ref}=6.4616$, 
$n_{ref}=1.1287$, $N_S=1440$, $N_P=400$, $\alpha^\prime_T=0.01$. 

\break
The relation between these parameters, the reference temperature and irradiance  and the vector $a$ of the IV characteristic \eqref{I} is \cite{tian2012, carrasco2014}
\begin{align}
a_1&=N_PI_{irr,ref}\left({G\over G_{ref}}\right)(1+\alpha^\prime_T(T-T_{ref})),\\
a_2&=N_PI_{0,ref}\left({T\over T{ref}}\right)^3exp\left({E_{g,ref}q\over kT_{ref}}-{E_{g}q\over kT}\right),\\
a_3&={q\over N_Sn_{ref}kT},\quad
a_4=R_{S,ref}{N_S\over N_P},\quad
a_5={N_P\over N_SR_{P,ref}}\left({G_{ref}\over G}\right),
\end{align}
where
$$
E_g=E_g(T)=1.16-4.73\times10^{-4}\times{T\over T+636}~(eV),\quad E_{g,ref}=E_g(T_{ref}).
$$

For the reference temperature $T=308.82^{\circ}C$ and irradiance $G=967.71~W/m^2$ these parameters become $a_1 = 726.21$, $a_2 = 5.9880\times 10^{-6}$, $a_3 = 0.0231$, $a_4 = 0.0732$, $a_5 = 0.0322$. The MPP voltage is $V_{*}=635.2~V$.

The estimated parameters $\hat a$ were generated combining the estimator  of the parameters $\theta$ given in \eqref{decest} with the mapping $\calg$ \eqref{G_a1}--\eqref{G_a345} of Proposition \ref{pro6.1}. 
Since $\calg$ has a potential division by zero a standard projection must be added to avoid singularities---this was, however, not needed in our case. For the filters \eqref{xi_123}, \eqref{xi_45} we use $\lambda=100$. To construct the regressor matrix $M_e$ in \eqref{yama} the delay signals \eqref{d_operator} were computed with $d_1=0.1$, $d_2=0.2$, $d_3=0.3$, $d_4=0.4$ and the scaling factor $\beta=1.25\times 10^{-3}$.

The control signal was taken as 
$$
u(t)=0.8+0.1\,\sin(3t)+0.1\,\sin(4t),
$$
which is ``rich enough" to verify the excitation condition $\Delta(t) \notin {\call}_2$ of Proposition \ref{pro5.1}. 

In Fig. \ref{fig4} the profiles of current, voltage and control  are shown. 
Fig. \ref{fig6} shows the evolution of the resulting parameter and MPP voltage estimation errors for different initial conditions.  Notice that the parameter estimates remain in a quiescent period in the time interval $t \in [0,2.5]$ and then move swiftly towards their correct value after sufficient excitation has been gathered with the delayed signals by the DREM estimator. From Fig. \ref{fig4} we see that this time interval coincides, as expected, with the peak values of current and voltage.

Similar plots are shown in Fig. \ref{fig5} for some fixed initial conditions and different adaptation gains.
As seen from the figures increasing $\gamma$ reduces the parameter convergence time, with the same role played by $\gamma_V$ on the estimation of the MPP voltage---a scenario consistent with the theory.  The choice of the parameter $\beta$ was quite critical and done on a trial-and-error basis to improve the condition number of the matrix $M_e$.

To check the tracking capabilities of the algorithms the simulations were repeated with a {\em time-varying temperature} $T$ which starts at $T_{ref}$ and linearly increases by $4^{\circ}C$ for $100$ seconds. The results are shown in Fig. \ref{fig7}---attracting the readers attention to the different time scale. Interestingly, the time variations of the true parameters induces excitation radically changing the transient performance of the parameter estimates. Notice also from Fig. \ref{fig7} (c) that, although the MPP voltage changes, the estimate $\hat V_*$ perfectly tracks this variation.   

Finally, Fig. \ref{fig8} shows the case where the  temperature $T$ linearly increases by $6^{\circ}C$  {\em and the irradiance $G$ increases by $5~W/m^2$} for $100$ seconds. Both starting from their reference values. The same behaviour pattern as in the previous scenario is repeated.

\section{Conclusions}
\lab{sec9}

Using the DREM technique a solution to the problem of on--line tracking of the MPP voltage of a PV system was proposed. The task is translated into a problem of parameter estimation of the classical five--parameter model of the system's IV characteristic. Given the large number of unknown parameters, its highly nonlinear dependence on the system dynamics and the scarce availability of measurement this is a complicated task that requires strict excitation conditions. The use of DREM allows us to relax the latter. Moreover, with the introduction of a suitable reparameterisation the nonlinear estimation problem is translated into a linear one. A numerical example with realistic simulations demonstrates that the estimator works properly under suitable excitation conditions.

As usual, the convergence analysis relies on the assumption that the parameters are constant, which implies that the temperature and the solar irradiance are constant. However, it is well known \cite{lju}, that on--line estimators are able to track slowly time--varying parameters, which is the scenario of practical interest for this application. 

Interestingly, it is shown in simulations that the time variations of the true parameters due to the temperature and irradiance changes generates additional excitation to the parameter estimator accelerating its transient behaviour. An open question of practical interest is whether the normal operation of the PV array would provide the excitation required by the estimator and, if not, how to inject small amplitude signals to the control that would ensure this condition without affecting the regulation quality. This situation is similar to the extremum seeking algorithms that rely on generating an oscillating regime around the MPP.    


\begin{figure*}[ht]
\centering
	\subcaptionbox{\label{fig61} $\|a-\hat a(t)\|$}{\includegraphics[width=0.32\textwidth]{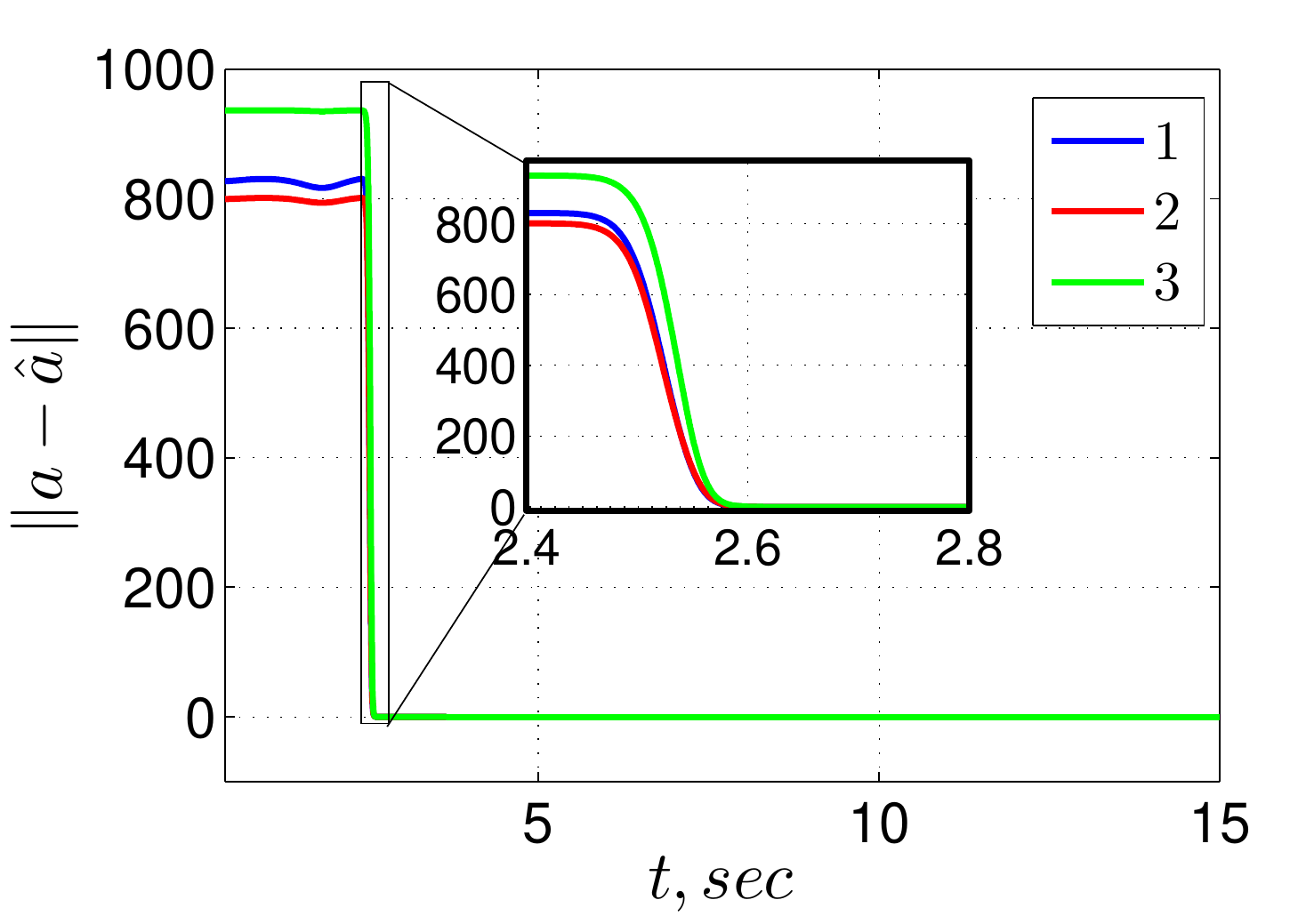}}
	\
	\subcaptionbox{\label{fig62} $V_*-\hat V_{*}(t)$}{\includegraphics[width=0.32\textwidth]{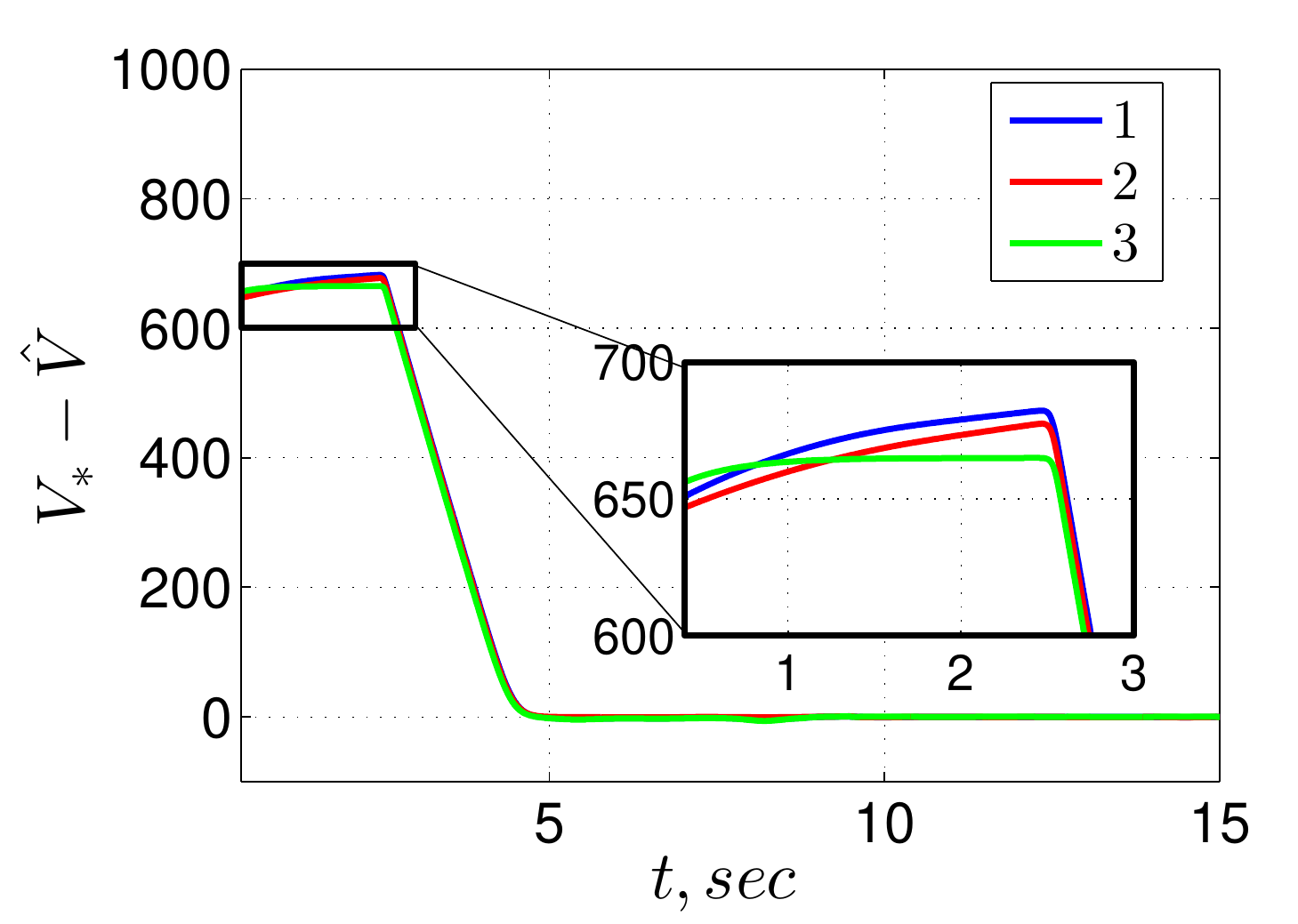}}
	\
	\subcaptionbox{\label{fig63} $V_*-\hat V_{*}(t)$}{\includegraphics[width=0.32\textwidth]{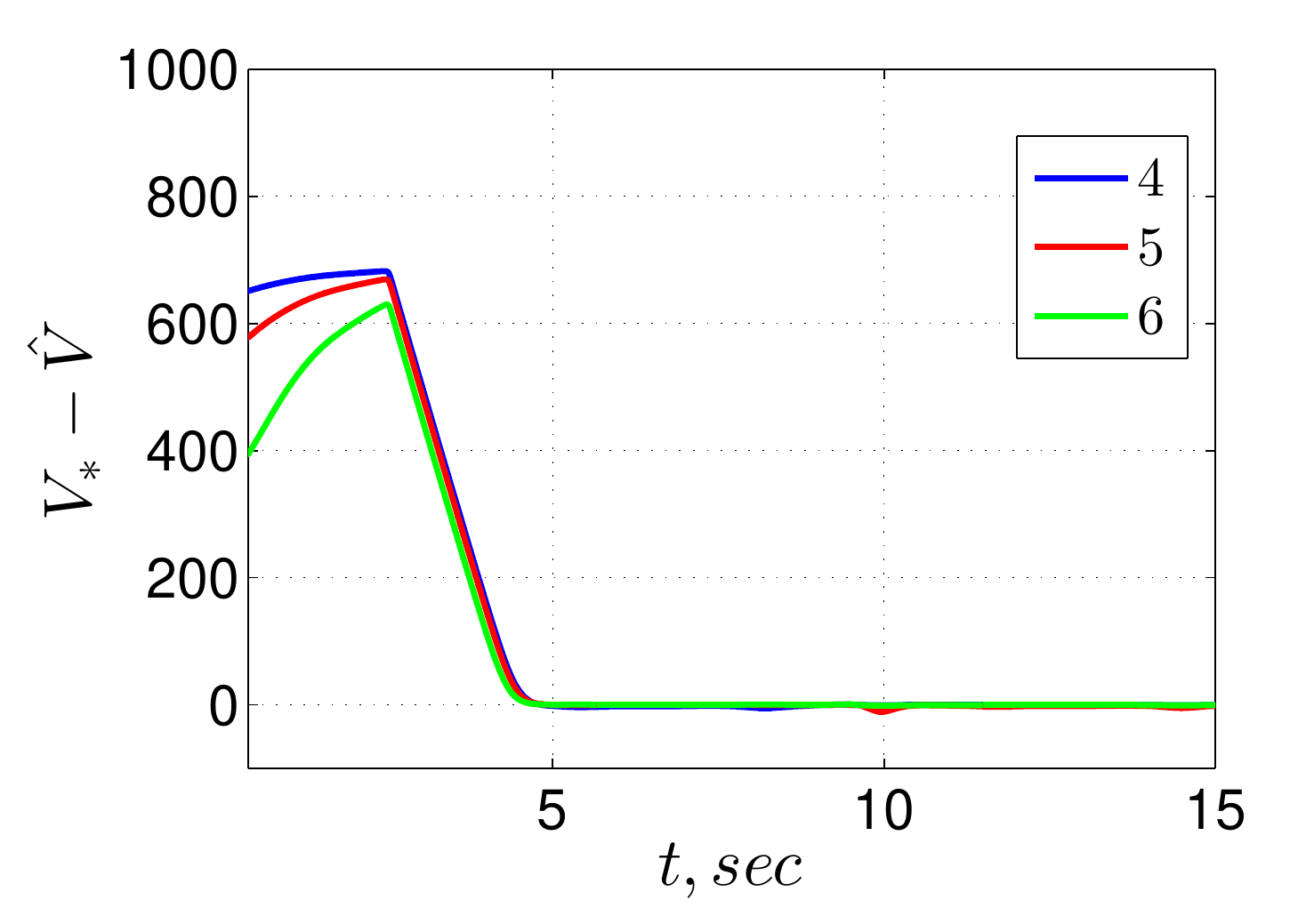}}

	\caption{\label{fig6} Estimation errors with  $\gamma=20$, $\gamma_V=0.5$ and different initial conditions:\\ 
	$1.$ (blue line) $\hat\theta(0)=(0.01,0.006,0.009,0.001)$, $\hat V_{*}(0)=0$;
	\\
$2.$ (red line) $\hat\theta(0)=(0.01,0.004,0.006,0.002)$, $\hat V_{*}(0)=0$;\\ 
$3.$ (green line) $\hat\theta(0)=(0.02,0.04,0.06,0.001)$, $\hat V_{*}(0)=0$;\\
$4.$ (blue line) $\hat\theta(0)=(0.01,0.006,0.009,0.001)$, $\hat V_{*}(0)=0$;
	\\
$5.$ (red line) $\hat\theta(0)=(0.01,0.006,0.009,0.001)$, $\hat V_{*}(0)=100$;\\ 
$6.$ (green line) $\hat\theta(0)=(0.01,0.006,0.009,0.001)$, $\hat V_{*}(0)=300$.}
\end{figure*}

\begin{figure*}[ht]
\centering
	\subcaptionbox{\label{fig51} $\|a-\hat a(t)\|$}{\includegraphics[width=0.32\textwidth]{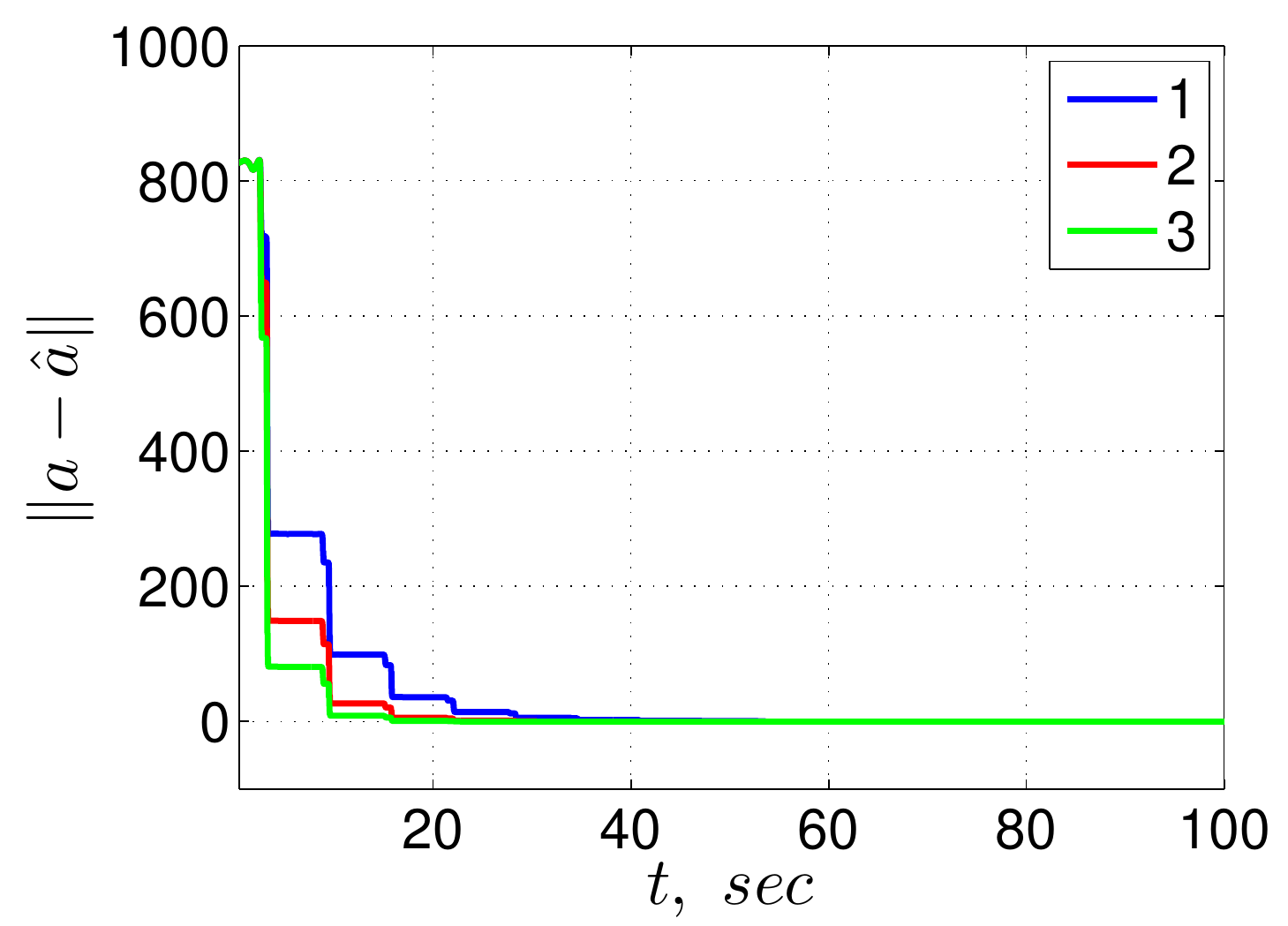}}
	\
	\subcaptionbox{\label{fig52} $V_*-\hat V_{*}(t)$}{\includegraphics[width=0.32\textwidth]{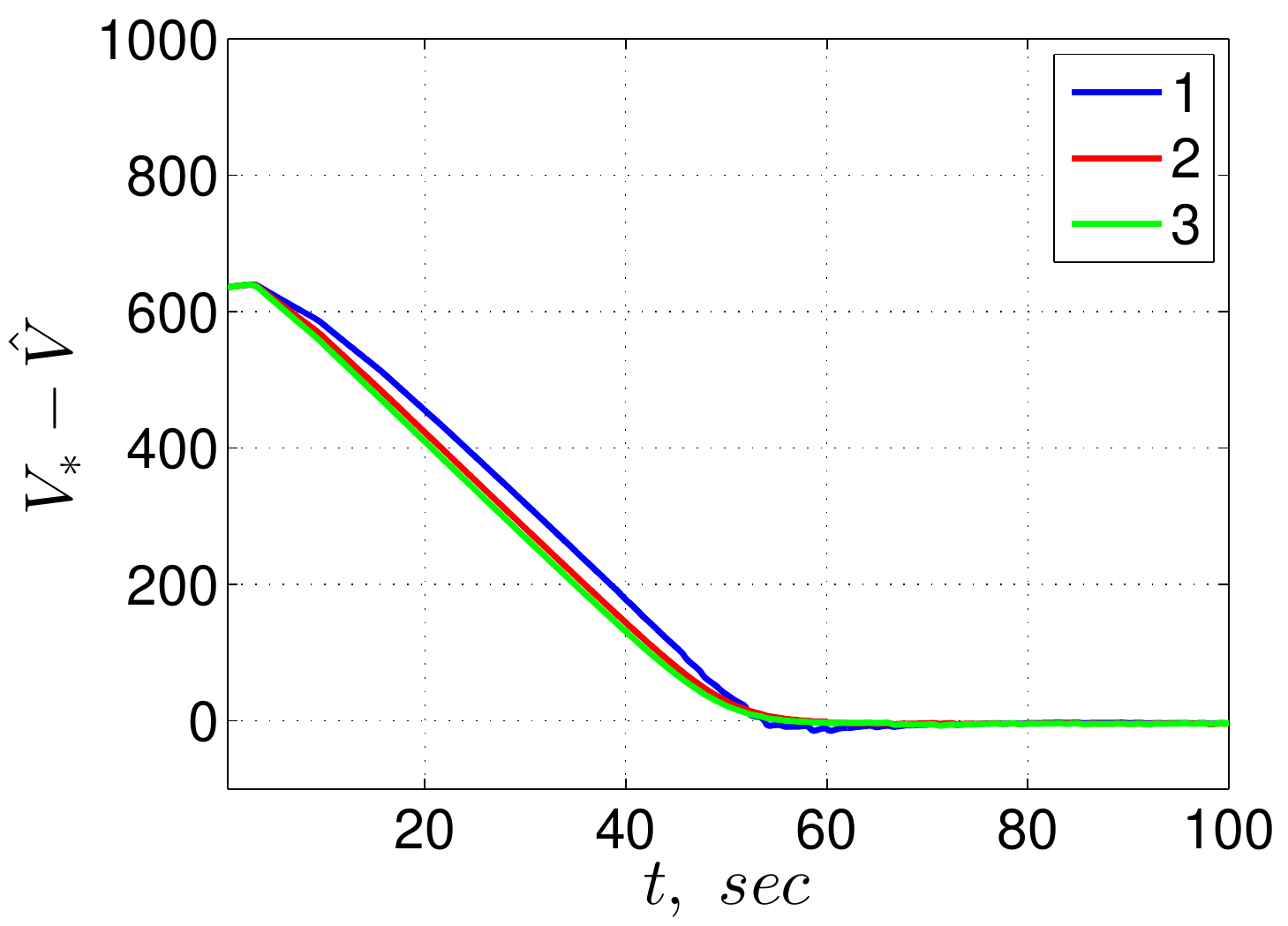}}
	\
	\subcaptionbox{\label{fig53} $V_*-\hat V_{*}(t)$}{\includegraphics[width=0.32\textwidth]{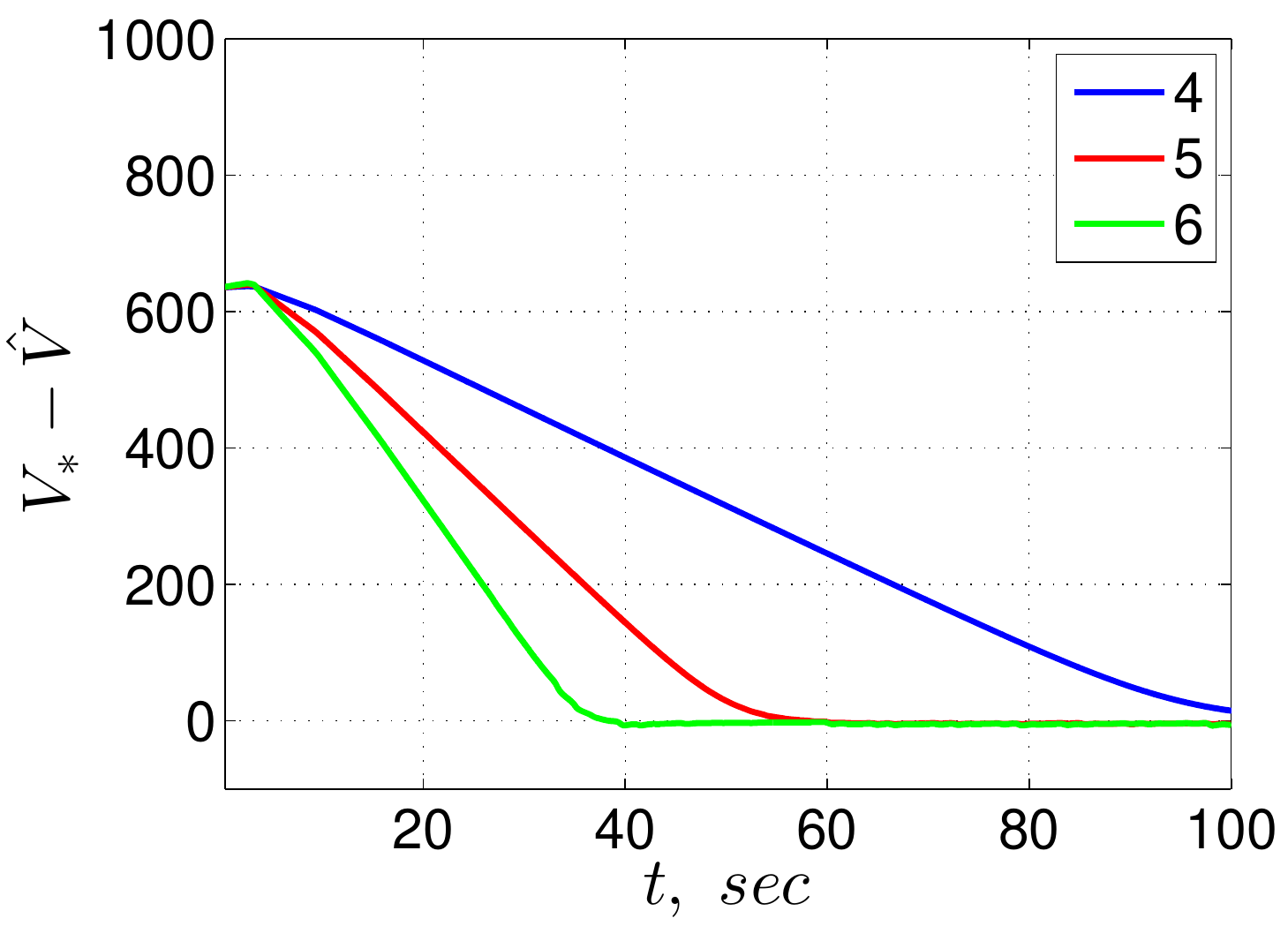}}
	
	\caption{\label{fig5} Estimation errors with initial conditions $\hat\theta(0)=(0.01,0.004,0.006,0.002)$ and different adaptation gains: \\ 
	$1.$ (blue line) $\gamma=0.3$, $\gamma_V=0.02$;\\ 
	$2.$ (red line) $\gamma=0.5$, $\gamma_V=0.02$;\\ 
	$3.$ (green line) $\gamma=0.7$, $\gamma_V=0.02$; \\
	$4.$ (blue line) $\gamma=0.5$, $\gamma_V=0.01$;\\ 
	$5.$ (red line) $\gamma=0.5$, $\gamma_V=0.02$;\\ 
	$6.$ (green line) $\gamma=0.5$, $\gamma_V=0.03$.}
\end{figure*}

\begin{figure*}[ht]
\centering
	\subcaptionbox{\label{fig71} $\|a-\hat a(t)\|$ }{\includegraphics[width=0.32\textwidth]{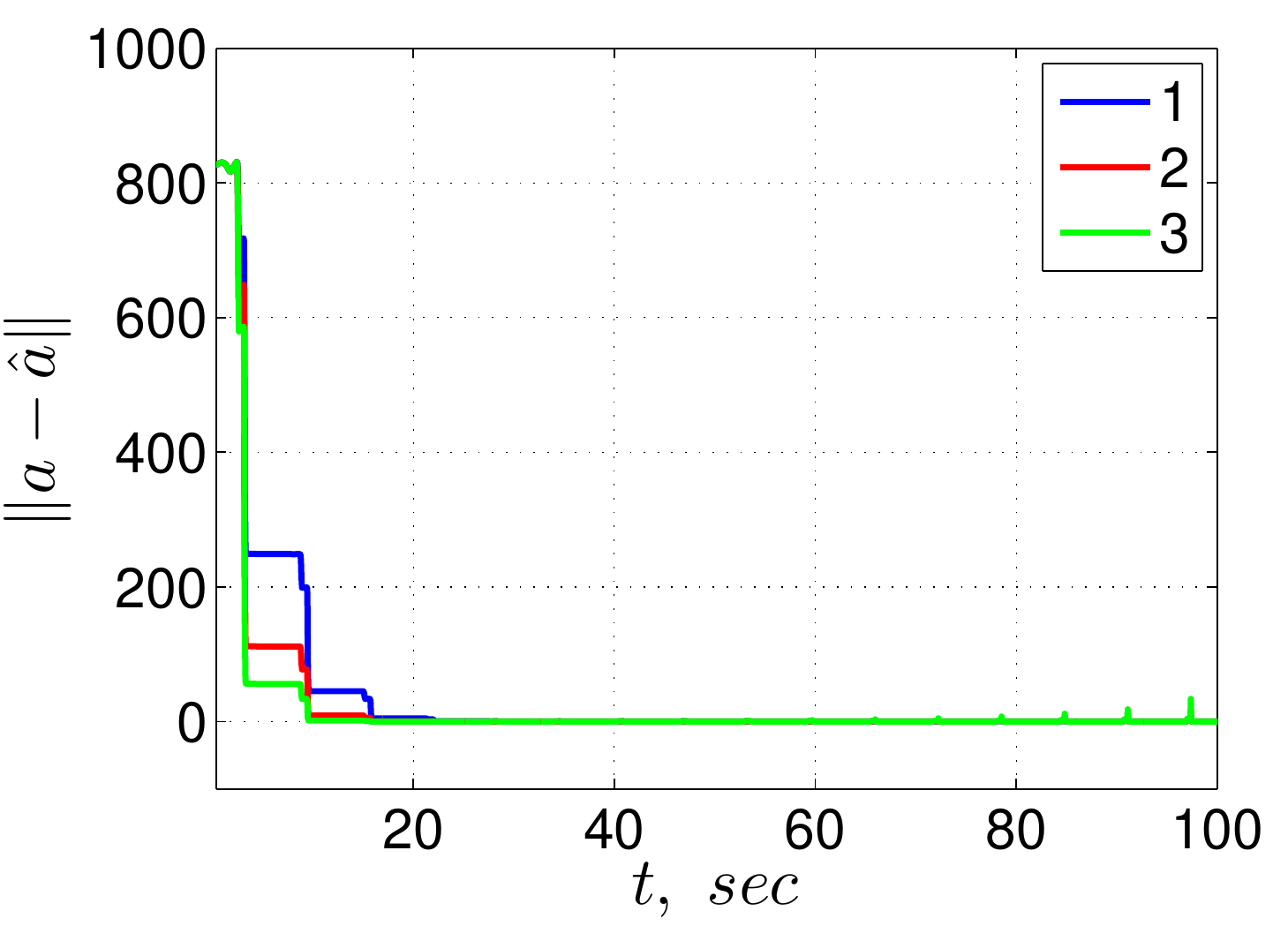}}
	\
	\subcaptionbox{\label{fig72} $V_*-\hat V_{*}(t)$ }{\includegraphics[width=0.32\textwidth]{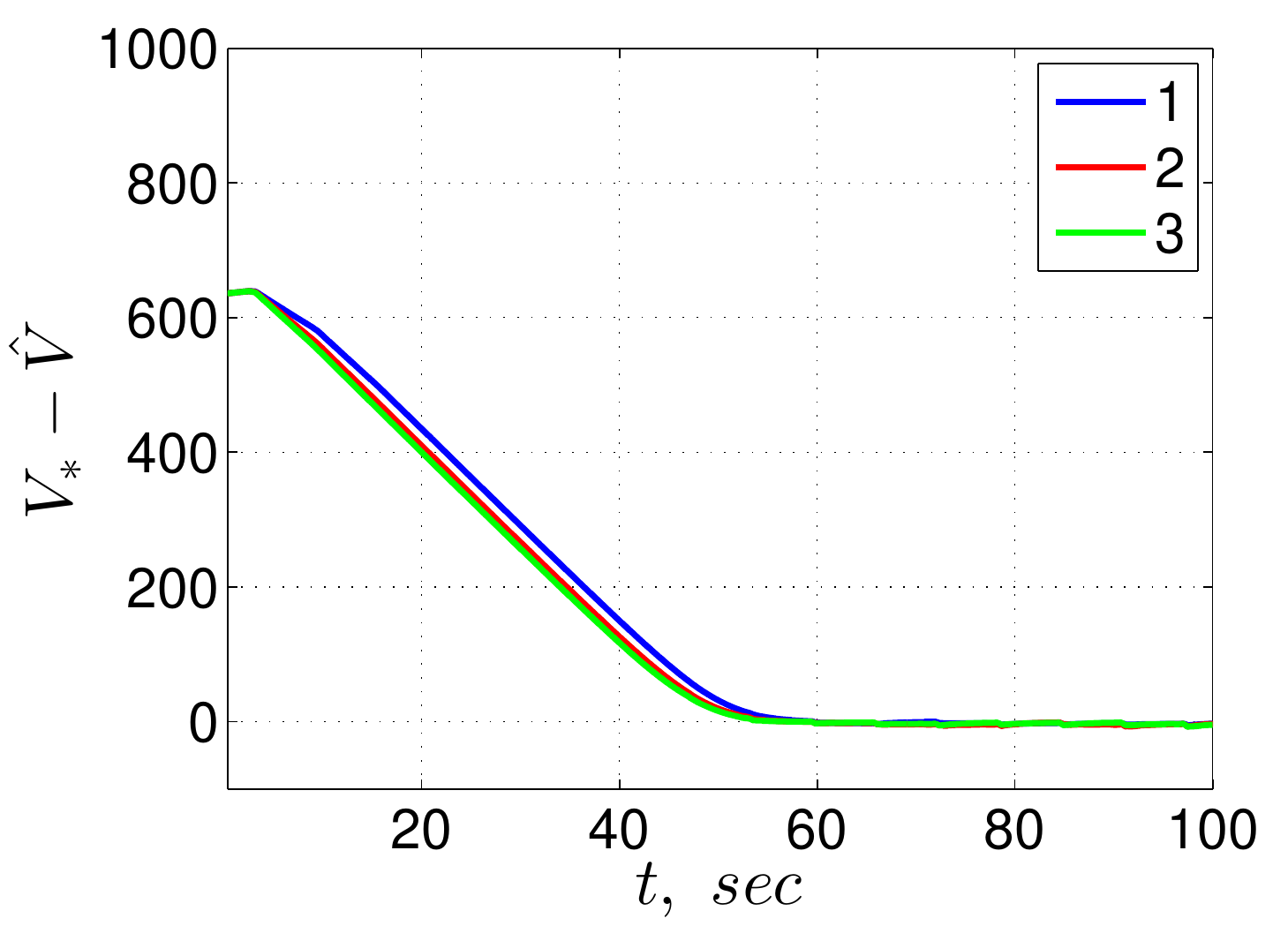}}
	\
	\subcaptionbox{\label{fig73} MPP voltage tracking}{\includegraphics[width=0.32\textwidth]{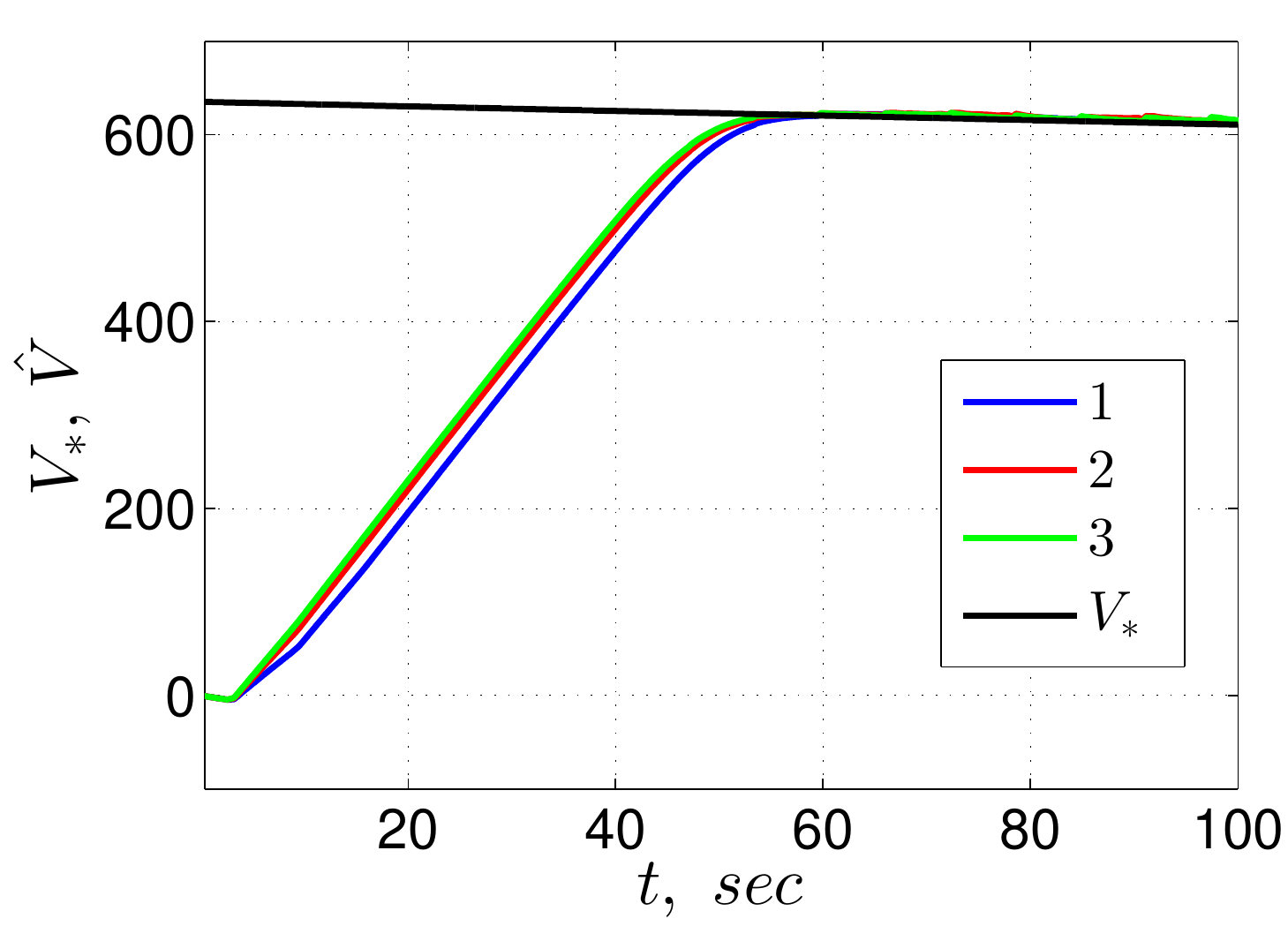}}
	
	\caption{\label{fig7} Estimation errors with $\hat\theta(0)=(0.01,0.004,0.006,0.002)$, $\gamma_V=0.02$ and different adaptation gains for time--varying temperature. \\ 
	$1.$ (blue line) $\gamma=0.3$;\\ 
	$2.$ (red line) $\gamma=0.5$;\\ 
	$3.$ (green line) $\gamma=0.7$.}
\end{figure*}

\begin{figure*}[ht]
\centering
	\subcaptionbox{\label{fig81}  $\|a-\hat a(t)\|$ }{\includegraphics[width=0.32\textwidth]{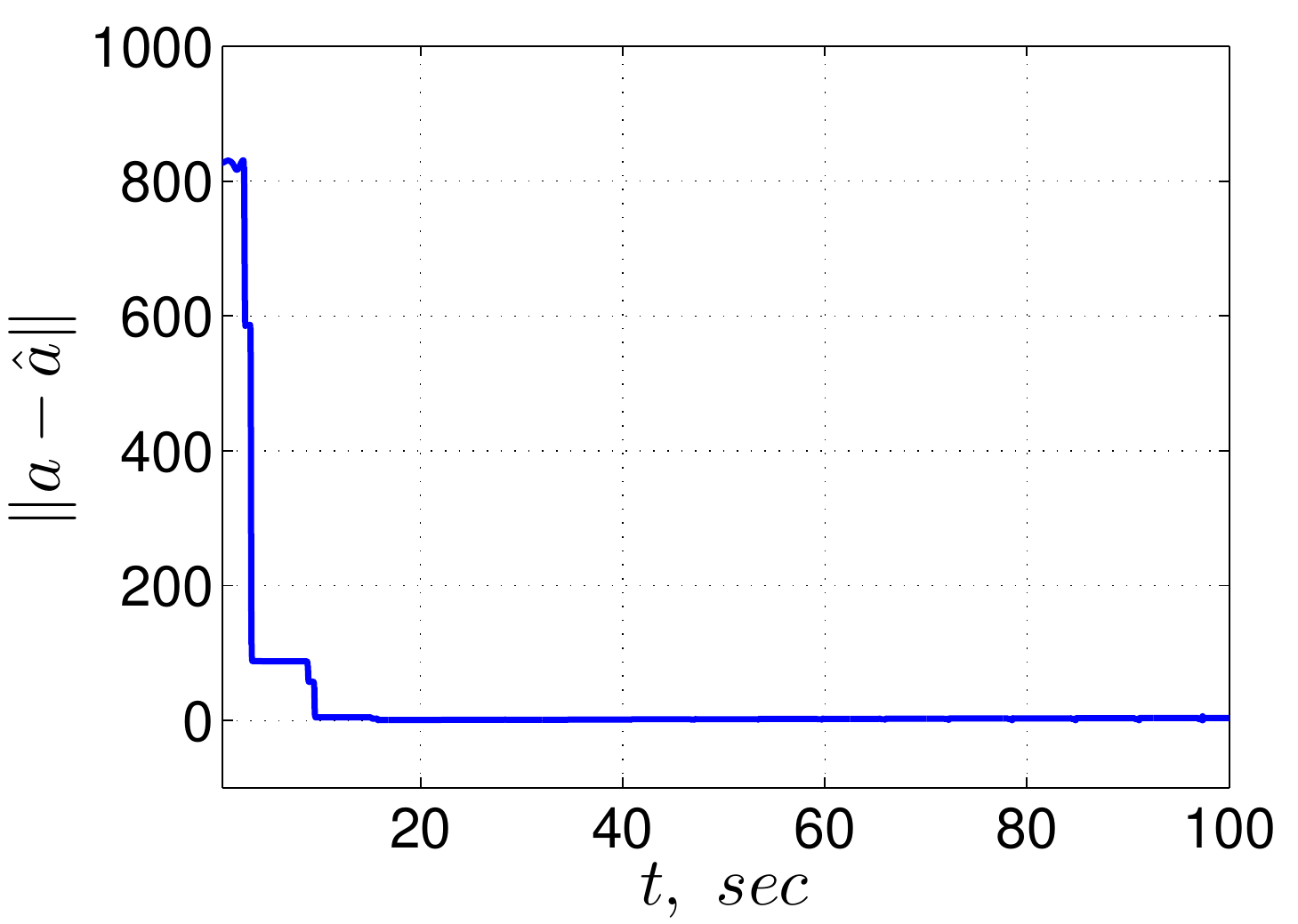}}
	\
	\subcaptionbox{\label{fig82}  $V_{*}-\hat V_{*}(t)$}{\includegraphics[width=0.32\textwidth]{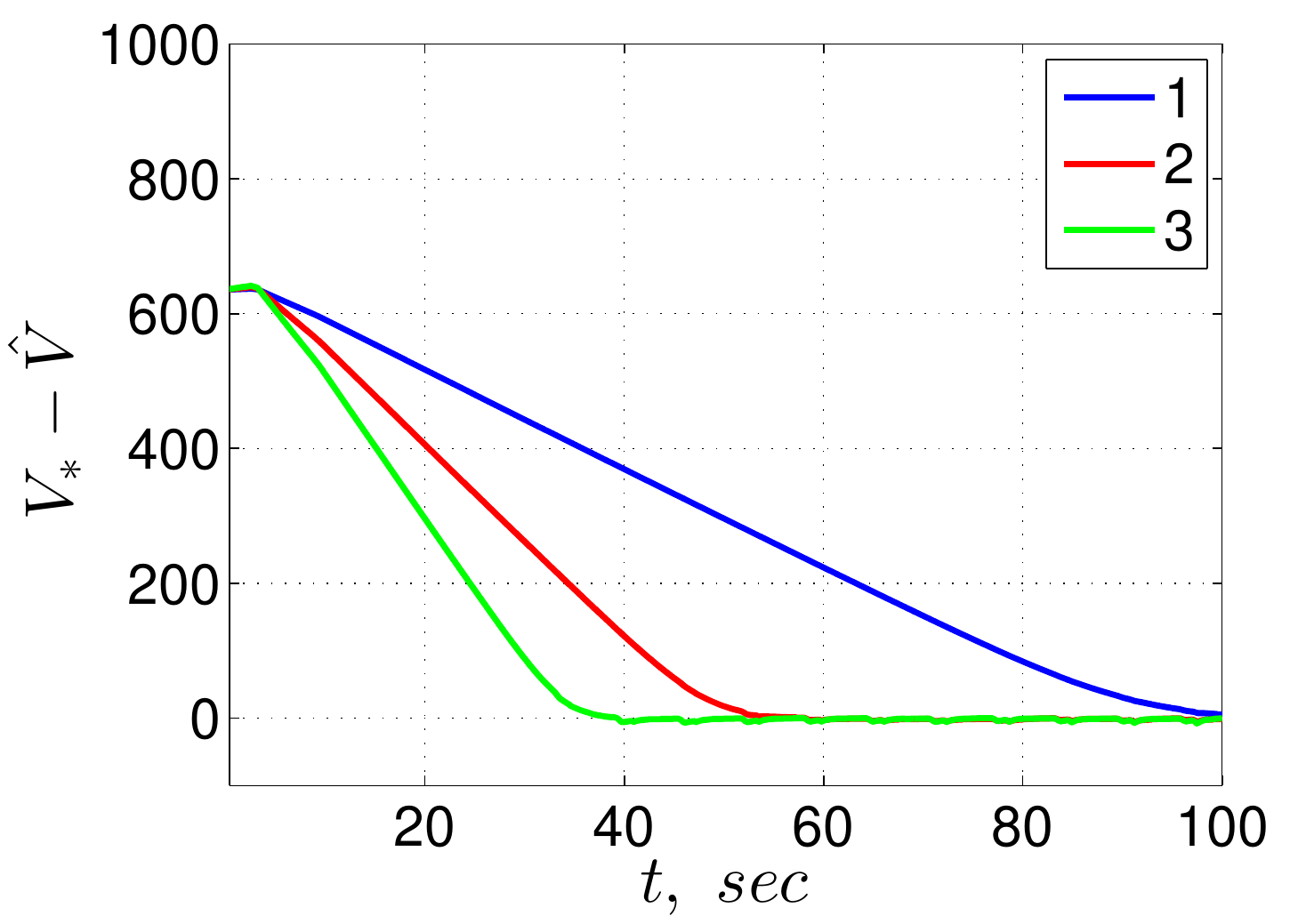}}
	\
	\subcaptionbox{\label{fig83} MPP voltage tracking}{\includegraphics[width=0.32\textwidth]{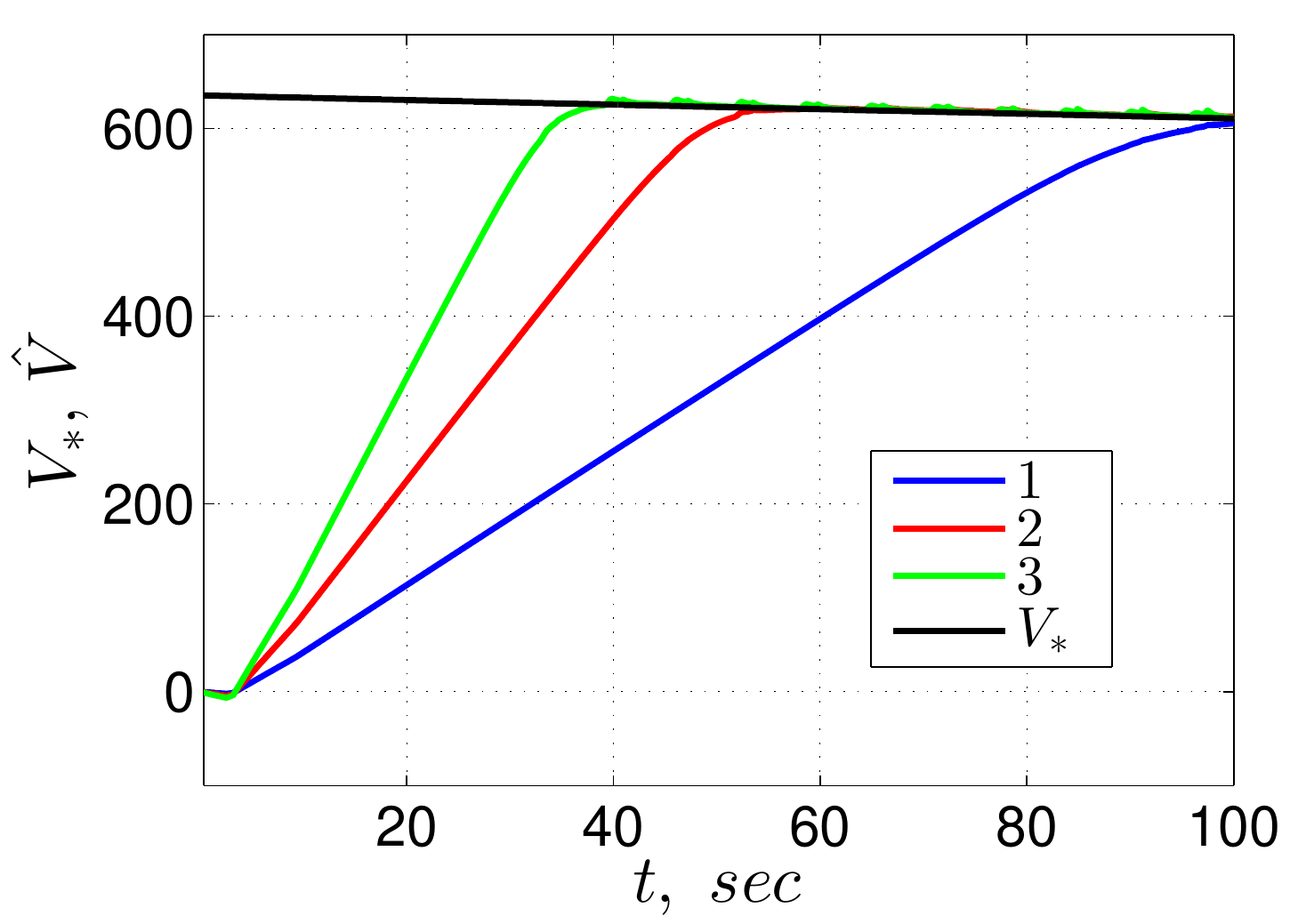}}
	
	\caption{\label{fig8} Estimation errors with $\hat\theta(0)=(0.01,0.004,0.006,0.002)$, $\gamma=0.5$ and different adaptation gains  for time--varying temperature and irradiance. \\ $1.$ (blue line) $\gamma_V=0.01$;\\ $2.$ (red line) $\gamma_V=0.02$;\\ $3.$ (green line) $\gamma_V=0.03$.}
\end{figure*}

\end{document}